\documentclass[11pt]{amsart}

\usepackage{amsmath}
\usepackage{amssymb}
\usepackage{amsthm}
\usepackage{color}
\usepackage{color,latexsym,amsfonts}
\usepackage{float}
\usepackage{ae,aecompl,graphicx}

\newtheorem{thm}{Theorem}[section]
\newtheorem{prop}{Proposition}[section]

\newtheorem{rem}{Remark}[section]

\def\<{\left<}\def\>{\right>}
\def\({\left(}\def\){\right)}

\allowdisplaybreaks

\begin{document}

\title{Optimal implementation delay of taxation with trade-off for L\'{e}vy risk Processes}
\author{Wenyuan Wang, Xueyuan Wu and Cheng Chi}

\address{Wenyuan Wang: School of Mathematical Sciences, Xiamen University, Fujian 361005, People's Republic of China. Email address: wwywang@xmu.edu.cn (W. Wang).
Xueyuan Wu: Corresponding author.
              The University of Melbourne
              VIC 3010, Australia.
              Email address: xueyuanw@unimelb.edu.au (X. Wu).
Cheng Chi: Guanghua School of Management,Peking University, Beijing 100871, People's Republic of China. Email address: pkuchicheng@pku.edu.cn (C. Chi).}


\subjclass[2000]{Primary: 60G51; Secondary:  91B30, 93E20}
\keywords{Spectrally negative L\'{e}vy process, General tax structure, First crossing time, Joint Laplace transform, Potential measure.}

\begin{abstract}
In this paper we consider two problems on optimal implementation delay of taxation with trade-off for spectrally negative L\'{e}vy insurance risk processes. In the first case, we assume that an insurance company starts to pay tax when its surplus reaches a certain level $b$ and at the termination time of the business there is a terminal value incurred to the company. The total expected discounted value of tax payments plus the terminal value is maximized to obtain the optimal implementation level $b^*$. In the second case, the company still pays tax subject to an implementation level $a$ but with capital injections to prevent bankruptcy. The total expected discounted value of tax payments minus the capital injection costs is maximized to obtain the optimal implementation level $a^*$. Numerical examples are also given to illustrate the main results in this paper.

\end{abstract}

\maketitle

\medskip

\section{Introduction }

In the recent actuarial science literature, there have been a lot of papers examining the impact of loss-carry-forward taxation under a risk process framework. Under the Cram\'{e}r-Lundberg risk models, \cite{Albrecher2008a}, \cite{Albrecher2009}, \cite{Albrecher2011}, \cite{Albrecher2007}, \cite{Cheung2012}, \cite{Ming2010}, \cite{Wang2011}, \cite{Wang2010}, and \cite{Wei2009} discussed a variety of tax-related problems. \cite{Wei2010} considered a Markov-modulated risk model with tax, and \cite{Albrecher2014a} discussed the tax identity problem for the Markov additive risk processes. Moreover, \cite{Albrecher2008a} studied a dual risk model with taxation, and \cite{Li2013} investigated a taxed time-homogeneous diffusion risk process. Among the above mentioned literatures, \cite{Albrecher2007} firstly addressed the problem of optimal implementation delay of taxation which aims to characterize the critical starting level for taxation that maximizes the expected aggregate discounted tax payments. Following the work of \cite{Albrecher2007}, \cite{Albrecher2008a} and \cite{Wang2010} solved the taxation implementation delay optimization problem, respectively, for the dual risk model with exponential jumps and the Cram\'{e}r-Lundberg risk model with a constant interest rate. In addition, \cite{Cheung2012} derived a necessary condition for the critical starting-tax surplus level of the optimal taxation implementation delay problem for the Cram\'{e}r-Lundberg risk model with surplus-dependent premium and tax rates.
\smallskip

Meanwhile, \cite{Albrecher2008b} solved the two-sided exit problem, the survival probability and the arbitrary moments of the discounted accumulated tax payments for the general spectrally negative L\'{e}vy risk process with loss-carry-forward tax payments collected at a constant rate. The taxation implementation delay problem of characterizing the optimal starting-tax surplus level to maximize the expected discounted total tax payments was also addressed. The results of \cite{Albrecher2008b} were generalized in \cite{Avram2017} by replacing the ruin stopping with the draw-down stopping. \cite{Kyprianou2009} considered a L\'{e}vy insurance risk model with tax payments of a more general structure. \cite{Kyprianou2012} further generalized the tax structure of \cite{Kyprianou2009} such that the tax rate function ranges over the whole real line.
Other research papers that studied L\'{e}vy risk processes with taxation include \cite{Albrecher2014b}, \cite{Hao2009}, \cite{Li2013}, \cite{Renaud2009}, \cite{Wang2012}, and \cite{Zhang2017}.

In investment literature, the terminal value of a company can be defined as the value of the business beyond the forecast period when future cash flows can be estimated. The discounted cash flow approach can be used to study the terminal value. In actuarial science literature, the terminal value is used as a tool to take account of the ruin time when studying problems like total present value of dividends until ruin, which has a retrospective focus of the concept. \cite{Thonhauser2007} introduced a value function which considers both expected dividends and the time value of ruin. The optimal strategy was identified for both the diffusion model and the Cram\'er-Lundberg model with exponential claim sizes. Meanwhile, \cite{Loeffen2009} considered an optimal dividends problem with a terminal value for the spectrally negative L\'{e}vy processes with a completely monotone jump density. \smallskip

This paper aims to study an optimal implementation delay of taxation under a L\'{e}vy insurance risk framework. We assume that taxation on an insurance company commences when the surplus of the insurance company reaches a predetermined level $b>0$, and we employ the terminal value concept to describe a particular feature of this taxation model: a lump sum incurred at termination of the insurance business, denoted by $S$, which displays a time-delaying feature. This lump sum could either be an income ($S<0$, a tax benefit) or an expense ($S>0$, a clawback of early tax relief) from the insurance company point of view. The termination is controlled by some predetermined condition, eg the surplus level down-crossing a certain level. In this paper, we assume that the risk process is terminated whenever its net surplus level becomes negative. We will discuss the optimal choice of a threshold surplus level, denoted be $b^*\in[0,\infty)$, at which the taxation commences, by maximizing the total expected discounted value of tax payments plus the terminal value. The existence of this terminal value $S$ causes an interesting trade-off between
maximizing the expected accumulated discounted tax payments until ruin and maximizing the expected timing value of $S$ at the termination of the insurance business, when determining the optimal tax threshold level $b^{*}$. In particular, compared to the scenario $S=0$ considered in \cite{Albrecher2008a}, \cite{Albrecher2007}, \cite{Albrecher2008b}, \cite{Avram2017}, \cite{Cheung2012} and \cite{Wang2010}, the inclusion of a negative terminal value $S\not= 0$ will lead to a trade-off between maximizing the total discounted tax payments and prolonging the life-span of the insurance company, rather than only attempting to collecting taxation as much and quickly as possible. Note that we may have no choice but to sacrifice the total discounted tax payments so that the life-span of the insurer can be prolonged.
\smallskip

In this paper, we also consider another type of trade-off in the above mentioned taxation model, with capital injections replacing the terminal value. The capital injections can prevent the insurance company from bankruptcy at some costs, i.e. the borrowing costs on those injected capitals. We assume that the insurance company can get certain tax deductions from the capital injection costs which results in a modified value function and slightly different optimization problem, with the initial tax implementation level $a$ retaking the center of the optimization. It needs to be pointed out that the introduction of capital injection results in a trade-off between maximizing the expected accumulated discounted tax payments and minimizing the expected accumulated discounted costs of capital injection, when determining the optimal tax threshold level $a^{*}$. Note that decreasing tax payments can help with reducing the need of capital injections as well as lowering the costs on receiving capital injections. Hence, on the contrary of the situation without capital injection as considered in \cite{Albrecher2008a}, \cite{Albrecher2007}, \cite{Albrecher2008b}, \cite{Avram2017}, \cite{Cheung2012}, and \cite{Wang2010}, the maximization of the total tax payments is no longer the sole objective when capital injection is in place. \cite{Albrecher2014b} studied an insurance surplus process subject to taxation and capital injection through a spectrally negative L\'{e}vy process which is refracted at its running maximum and at the same time reflected from below at a certain level. \cite{Zhou2007} also considered exit problems for spectrally negative L\'{e}vy processes reflected at either the supremum or the infimum.

This paper is organized as follows: Section 2 comprises some preliminaries concerning the spectrally negative  L\'{e}vy process. In Section 3 we present the mathematical setup of the taxation implementation delay problem with a terminal value. We derive the optimal tax implementation level $b^*$ given that the L\'{e}vy measure has a completely monotone density. The corresponding maximized value function is also determined at $b^*$. In Section 4 we study the optimal implementation delay of taxation with capital injections replacing the terminal value. The optimal tax implementation level $a^*$ is determined as well as the corresponding maximized value function. Also, in Section 5 we illustrate the optimal tax implementation results obtained in previous two sections using numerical examples.

\section{Preliminaries of the spectrally negative L\'evy process}\label{sec:2}

In this section, we briefly review some preliminaries of the spectrally negative L\'{e}vy process.
Let $X=\{X(t);t\geq0\}$ with probability laws $\{\mathbb{P}_{x};x\in \mathbb{R}\}$ and natural filtration $\{\mathcal{F}_{t};t\geq0\}$ be a spectrally negative L\'{e}vy process, with the usual exclusion of pure increasing linear drift and the negative of a subordinator. Denote by $\overline{X}$ and $\underline{X}$, i.e. $$\overline{X}(t)=\sup\limits_{0\leq s\leq t}X(s),\quad \underline{X}(t)=\inf\limits_{0\leq s\leq t}X(s),\quad t\geq0,$$ the running supremum and running infimum process of $X$, respectively.\smallskip

Define the Laplace exponent of $X$ by
\begin{eqnarray}
\Psi(\theta):=\log\mathbb{E}_{x}\Big[\mathrm{e}^{\theta (X(1)-x)}\Big],\nonumber
\end{eqnarray}
which is known to be finite for  $\theta\in[0,\infty)$ where it is strictly convex and infinitely differentiable.
As defined in \cite{Bertoin1996}, for each $q\geq0$, the scale function $W_{q}:\,[0,\infty)\rightarrow[0,\infty)$ is a unique strictly increasing and continuous function with Laplace transform
\begin{eqnarray}
\int_{0}^{\infty}\mathrm{e}^{-\theta x}W_{q}(x)\mathrm{d}x=\frac{1}{\Psi(\theta)-q},\quad \theta>\Phi(q),\nonumber
\end{eqnarray}
where $\Phi(q)$ is the largest solution of the equation $\Psi(\theta)=q$. For convenience, we extend the domain of $W_{q}$ to the whole real line by setting $W_{q}(x) = 0$ for all $x < 0$.  In particular, write $W=W_0$ for simplicity. When $X$ has sample paths of unbounded variation, or when $X$ has sample paths of bounded variation and the L\'{e}vy measure has no atoms, the scale function $W_{q}$ is continuously differentiable over $(0, \infty)$. See \cite{Chan2011} for more detailed discussions on the smoothness of scale functions. \smallskip

We define
\begin{eqnarray}
Z_{q}(x)=1+q\int_{0}^{x}W_{q}(z)\mathrm{d}z \quad \mbox{and} \quad \overline{Z}_{q}(x)=\int_{0}^{x}Z_{q}(z)\mathrm{d}z,\quad q\geq0, \,x\geq0,\nonumber
\end{eqnarray}
with $Z_{q}(x)=1$ for $x<0$. For the process $X$, let its first down-crossing time of level $0$ and first up-crossing time of level $b\in(0,\infty)$ be
\begin{eqnarray}
\tau_{0}^{-}:=\inf\{t\geq0: X(t)<0\}\,\,\,\text{and}\,\,\,\tau^+_{b}:=\inf\{t\geq0:X(t)>b\}.\nonumber
\end{eqnarray}
Then, Chapter 8 in \cite{Kyprianou2006} gives, for $q\geq0$ and $b\in(0,\infty)$,
\begin{eqnarray}
\label{exi.pro.ell=0}
\mathbb{E}_{x}\Big[e^{-q\tau^{+}_{b}}
\mathbf{1}_{\{\tau^{+}_{b}<\tau_{0}^{-}\}}\Big]
=\frac{W_{q}(x)}{W_{q}(b)},\quad x\in(0, b].
\end{eqnarray}

Let $x\wedge y$ denote $\min(x, y)$ and let $x\vee y$ denote $\max(x, y)$. We now define a process $\{Y(t);t\geq0\}$ as
$$Y(t):=X(t)-\underline{X}(t)\wedge 0,\quad t\geq 0,$$
which is the L\'evy process $X$ reflected at its infimum. In risk theory, if $X_t$ denotes the surplus level of an insurance company at time $t$, then the term
$-\underline{X}(t)\wedge 0$
could be used to represent the cumulative capital injection up to time $t$, and hence $Y_t$ is the surplus process with capital injections such that $Y(t)\geq0$ for $t\geq0$, i.e., ruin never occurs.
Let $$\rho_{a}^{+}:=\inf\{t\geq0;Y(t)>a\},\quad a\in(0,\infty),$$
which is the first up-crossing time of level $a$ for the risk process $Y$.
Then, for $q\geq0$ and $b\in(0,\infty)$, by Proposition 2 of \cite{Pistorius2004}, it holds that
\begin{equation}\label{two.side.exit.}
\mathrm{E}_x\big[\mathrm{e}^{-q\rho_{a}^{+}}\big]=Z_{q}(x)/Z_{q}(a),\quad x\in[0,a].
\end{equation}
In addition, by the proof of Theorem 1 of \cite{Avram2007} (see the 1st and 2nd blocks of equations on Page 167),  we have
\begin{equation}\label{exp.cap.inj.unt.exit.}
\mathbb{E}\Big[\int_{0}^{\rho_{a}^{+}}\mathrm{e}^{-q t}\mathrm{d}(-\underline{X}(t)\wedge0)\Big]
=-\frac{\Psi^{\prime}(0+)}{q}+\frac{\overline{Z}_{q}(a)+\frac{\Psi^{\prime}(0+)}{q}}{Z_{q}(a)},\quad a\in(0,\infty),
\end{equation}
which is the expected total discounted capital injection made from time 0 until the risk process $Y$ hits $a$.

By Remark 4.5 of \cite{Wang2018b}, we see that
the function $\frac{W_{q}(x)}{W_{q}^{\prime}(x)}$ is increasing with respect to $x$, i.e.,
\begin{eqnarray}
\label{W}
\hspace{-0.3cm}&&\hspace{-0.3cm}
W_{q}(x)W_{q}^{\prime\prime}(x)<\left(W_{q}^{\prime}(x)\right)^{2},\quad x\in(0,\infty).
\end{eqnarray}
From Proposition 2 (ii) of \cite{Pistorius2004}, one has
\begin{eqnarray}
\label{25.ax}
&&
\mathrm{E}_{0}\big[\mathrm{e}^{-q\overline{\tau}_{a}^{+}}\big]=
Z_{q}(a)-
\frac{q\left(W_{q}(a)\right)^{2}}{W_{q}^{\prime}(a)}
,\nonumber
\end{eqnarray}
where
$$\overline{\tau}_{a}^{+}:=\inf\{t\geq0; \overline{X}(s)\vee0-X(t)>a\},$$
is the first passage time of $X$ reflected at its supremum.
Because $\lim\limits_{a\rightarrow \infty}\overline{\tau}_{a}^{+}=\infty$, one knows
\begin{eqnarray}
\label{phi.0}
\hspace{-0.3cm}&&\hspace{-0.3cm}
\lim_{a \to \infty}\Big(Z_{q}(a)-q\frac{\left(W_{q}(a)\right)^{2}}{W_{q}^{\prime}(a)}\Big)=0.
\end{eqnarray}
In addition,
it is found in \cite{Zhou2007}
that
\begin{eqnarray}
\label{phi}
\hspace{-0.3cm}&&\hspace{-0.3cm}
\lim_{a \to \infty}\frac{W_{q}^{\prime}(a)}{W_{q}(a)}=\Phi(q),
\end{eqnarray}
which together with the L'H\^opital's rule gives
\begin{eqnarray}
\label{phi.00}\lim\limits_{a\rightarrow \infty}\frac{Z_{q}(a)}{Z_{q}^{\prime}(a)}=
\frac{1}{\Phi(q)}.
\end{eqnarray}

\section{Optimal delay of taxation implementation with a terminal value at ruin}

We assume that the before-tax surplus process for an insurance company evolves according to $X$ defined above with $X(0)=x$. In the actuarial science literature, such a surplus process is referred to as a L\'evy risk process. For any given time $t\geq0$,
the insurer pays tax at a constant rate $\ell\in[0,1)$, provided that the insurer is in a profitable situation at time $t$, which is monitored by its running supremum process $\overline{X}(t)$.
Hence, the cumulative taxation collected over the time interval $[0,t]$, is given by
\begin{eqnarray}
\ell\left(\overline{X}(t)-x\right),\quad t\geq 0,\nonumber
\end{eqnarray}
and then the net surplus after tax at time $t$ is governed by
\begin{eqnarray}\label{2U}
U(t)=X(t)-\ell\left(\overline{X}(t)-x\right),\quad t\geq 0,
\end{eqnarray}
with $U(0)=x$.
A risk process with taxation similar to (\ref{2U}) was first considered in \cite{Albrecher2007} where $X$ was a Compound Poisson process with drift, and \cite{Albrecher2008b} extended to a general spectrally negative L\'evy process $X$. In those two pioneer papers \cite{Albrecher2007} and \cite{Albrecher2008b}, the arbitrary moments of the accumulated discounted tax payments and the two-sided exit problem were solved, and criteria for the surplus level for starting taxation to maximize the expected accumulated discounted tax payments were also characterized.

Define the first down-crossing time of $U$ at level $0$ and the first up-crossing time of $U$ at level $b$ as
\begin{eqnarray}
\sigma_{0}^{-}:=\inf\{t\geq0:U(t)<0\}\,\,\,\text{and}\,\,\,\,\sigma_{b}^{+}:=\inf\{t\geq0:U(t)>b\},\nonumber
\end{eqnarray}
with the usual convention $\inf\emptyset:=\infty$.

The following results of the two-sided exit problem, the Laplace transform of $\sigma_{0}^{-}$ and the expected accumulated discounted tax payments paid until $\sigma_{0}^{-}$ can be found in  \cite{Kyprianou2012}, with a slight adaptation, for $b\in(0,\infty)$, $x\in(0, b]$:
\begin{eqnarray}
\label{exi.pro.}
\mathbb{E}_{x}\Big[\mathrm{e}^{-q\sigma^{+}_{b}}
\mathbf{1}_{\{\sigma^{+}_{b}<\sigma_{0}^{-}\}}\Big]
=\left(\frac{W_{q}(x)}{W_{q}(b)}\right)^{\frac{1}{1-\ell}},
\end{eqnarray}
\begin{eqnarray}
\label{joint.lapl.for.U.}
\mathbb{E}_{x}\Big[\mathrm{e}^{-q \sigma_{0}^{-}};\sigma_{0}^{-}<\sigma_{b}^{+}\Big]
\hspace{-0.3cm}&=&\hspace{-0.3cm}
\frac{1}{1-\ell}
\int_{x}^{b}
\left(\frac{W_q(x)}{W_q(z)}\right)^{\frac{1}{1-\ell}} \left(\frac{W_{q}^{\prime}(z)}{W_{q}(z)}Z_{q}(z)-qW_{q}(z)\right)\,\mathrm{d}z, \qquad q\geq0,
\end{eqnarray}
and
\begin{eqnarray}\label{e.d.a.t.}
\mathbb{E}_{x}\Big[\int_{0}^{\sigma_{b}^{+}
\wedge\sigma_{0}^{-}}\mathrm{e}^{-q t}\ell\,\mathrm{d}\overline{X}(t)\Big]
\hspace{-0.3cm}&=&\hspace{-0.3cm}
\frac{\ell}{1-\ell}\int_x^{b}\left(\frac{W_{q}(x)}{W_{q}(z)}\right)^{\frac{1}{1-\ell}}\mathrm{d}z,
\qquad q\geq0.
\end{eqnarray}
By the arguments in Section 4.1 of \cite{Wang2019} and (7) in Lemma 3.2 of \cite{Wang2018}, one recalls that, for $x\in(0,\infty)$, $a\in(x,\infty)$,
\begin{equation}
\label{12}
\mathbb{E}_{x}\big[\mathrm{e}^{-q \sigma_{0}^{-}}\,\hbar\left(\overline{U}(\sigma_{0}^{-})\right); \sigma_{0}^{-}<\sigma_{a}^{+}\big]
=
\frac{1}{1-\ell}\int_{x}^{a}\hbar(s)
\left(\frac{W_{q}(x)}
{W_{q}(s)}\right)^{\frac{1}{1-\ell}}
\left(\frac{W_{q}^{\prime}(s)}{W_{q}(s)}Z_{q}(s)-qW_{q}(s)\right) \mathrm{d}s,
\end{equation}
for any measurable function $\hbar:\,(-\infty,\infty)\rightarrow(-\infty,\infty)$ such that the integration on the right-hand side of \eqref{12} is finite, and
\begin{eqnarray}\label{14}
\hspace{-0.3cm}&&\hspace{-0.3cm}
\mathbb{E}_{x}\big[\mathrm{e}^{-q \sigma_{0}^{-}}\left|U(\sigma_{0}^{-})\right|;\sigma_{0}^{-}<\sigma_{a}^{+}\big]
=\frac{1}{1-\ell}\int_{x}^{a}
\left(\frac{W_{q}(x)}
{W_{q}(s)}\right)^{\frac{1}{1-\ell}}
\nonumber\\
&&\qquad\qquad\qquad\qquad\qquad
\times
\left(Z_{q}(s)-\Psi^{\prime}(0+)W_{q}(s)-\frac{\overline{Z}_{q}(s)-\Psi^{\prime}(0+)
\overline{W}_{q}(s)}{W_{q}(s)}W_{q}^{\prime}(s)\right)
\mathrm{d}s.
\end{eqnarray}

Unlike the conventional problems studied on risk models with taxation, in this section we shall take into consideration a terminal value $S\in(-\infty,\infty)$ incurred to the insurance company at the ruin time.
By \eqref{joint.lapl.for.U.} and \eqref{e.d.a.t.}, the expected discounted tax payments made by the insurance company until ruin plus the expected discounted terminal value at the ruin time can be calculated as, for $q\ge 0$, $b=\infty$,
\begin{eqnarray}
\label{psi}
\psi(x)
\hspace{-0.2cm}&:=&\hspace{-0.2cm}
\mathbb{E}_x\Big[\int_{0}^{\sigma_{0}^{-}}\mathrm{e}^{-q t}\ell\,\mathrm{d}\overline{X}(t)+S\mathrm{e}^{-q \sigma_{0}^{-}}\Big]
\nonumber\\
\hspace{-0.3cm}&=&\hspace{-0.3cm}
\frac{S}{1-\ell}\int_x^{\infty}\left(\frac{W_{q}(x)}{W_{q}(z)}\right)^{\frac{1}{1-\ell}}
\left(\frac{W_{q}^{\prime}(z)}{W_{q}(z)}Z_{q}(z)-qW_{q}(z)\right)\mathrm{d}z
\nonumber\\
\hspace{-0.3cm}&&\hspace{-0.3cm}
+\frac{\ell}{1-\ell}\int_x^{\infty}\left(\frac{W_{q}(x)}{W_{q}(z)}\right)^{\frac{1}{1-\ell}}\mathrm{d}z,
\quad  x\in[0,\infty).
\nonumber
\end{eqnarray}

Inspired by \cite{Albrecher2007} and \cite{Albrecher2008b}, we presume a threshold level $b\in[0,\infty)$, which is the bar that the insurance company's surplus needs to hit before tax is payable. Then we have a modified net surplus process, for $b\in[0, \infty)$,
\begin{eqnarray}\label{U_b}
U_b(t)=X(t)-\ell\left(\overline{X}(t)\vee b-x\vee b\right),\quad t\geq 0,\nonumber
\end{eqnarray}
with $U_b(0)=x$. The objective of this section is to determine the optimal threshold level $b^*\in[0, \infty)$ from which on the tax authority should collect taxes in order to maximize the expected total discounted value of tax payments plus the terminal value at ruin. For this purpose, let
$$\sigma_{0}^{-}(b):=\inf\{t\geq0; \,U_b(t)<0\},$$
be the ruin time of $U_b$. In addition, define the value function $\phi(x;b)$ with implementation delay of taxation by
\begin{eqnarray}
\label{phi.exp.00}
\phi(x;b)
\hspace{-0.3cm}&:=&\hspace{-0.3cm}
\mathbb{E}_{x}\left[\mathbf{1}_{\{\tau^+_{b}<\sigma_{0}^{-}(b)\}}\int_{\tau^+_{b}}^{\sigma_{0}^{-}(b)}\mathrm{e}^{-q t}\ell\,\mathrm{d}\overline{X}(t)
+S\mathrm{e}^{-q\sigma_{0}^{-}(b)}
\right],\nonumber
\end{eqnarray}
which, roughly speaking, is the value function of the expected discounted tax payments until ruin plus the expected discounted terminal value at ruin where taxation is delayed until the moment the risk process $U_b$ hits $b$. Due to the strong Markov property, \eqref{exi.pro.ell=0}, \eqref{joint.lapl.for.U.} and \eqref{e.d.a.t.}, we have, for $b\in(0,\infty)$, $x\in(0,b]$,
\begin{eqnarray}
\label{phi.exp.}
\phi(x;b)
\hspace{-0.3cm}&=&\hspace{-0.3cm}
\mathbb{E}_{x}\Big[\mathbf{1}_{\{\tau^+_{b}<\tau_{0}^{-}\}}\int_{\tau^+_{b}}^{\sigma_{0}^{-}(b)}\mathrm{e}^{-q t}\ell\,\mathrm{d}\overline{X}(t)
+S\mathrm{e}^{-q \sigma_{0}^{-}(b)}\mathbf{1}_{\{\tau^+_{b}<\sigma_{0}^{-}(b)\}}
+S\mathrm{e}^{-q \tau_{0}^{-}}\mathbf{1}_{\{\tau_{0}^{-}<\tau^+_{b}\}}\Big]
\nonumber\\
\hspace{-0.3cm}&=&\hspace{-0.3cm}
\mathbb{E}_{x}\Big[\mathrm{e}^{-q \tau^+_{b}}\mathbf{1}_{\{\tau^+_{b}<\tau_{0}^{-}\}}\Big]
\mathbb{E}_{b}\Big[\int_{0}^{\sigma_{0}^{-}}\mathrm{e}^{-q t}\ell\,\mathrm{d}\overline{X}(t)\Big]
+S\,\mathbb{E}_{x}\Big[\mathrm{e}^{-q \tau^+_{b}}\mathbf{1}_{\{\tau^+_{b}<\sigma_{0}^{-}\}}\Big]
\mathbb{E}_{b}\Big[\mathrm{e}^{-q \sigma_{0}^{-}}\Big]
\nonumber\\
\hspace{-0.3cm}&&\hspace{-0.3cm}
+S\,\mathbb{E}_{x}\Big[\mathrm{e}^{-q \tau_{0}^{-}}\mathbf{1}_{\{\tau_{0}^{-}<\tau^+_{b}\}}\Big]
\nonumber\\
\hspace{-0.3cm}&=&\hspace{-0.3cm}
\frac{S}{1-\ell}\frac{W_{q}(x)}{W_{q}(b)}
\int_b^{\infty}\left(\frac{W_{q}(b)}{W_{q}(z)}\right)^{\frac{1}{1-\ell}}
\left(\frac{W_{q}^{\prime}(z)}{W_{q}(z)}Z_{q}(z)-qW_{q}(z)\right)\mathrm{d}z
\nonumber\\
\hspace{-0.3cm}&&\hspace{-0.3cm}
+S\left(Z_{q}(x)-\frac{W_{q}(x)}{W_{q}(b)}Z_{q}(b)\right)
+\frac{\ell}{1-\ell}\frac{W_{q}(x)}{W_{q}(b)}\int_b^{\infty}\left(\frac{W_{q}(b)}{W_{q}(z)}\right)^{\frac{1}{1-\ell}}\mathrm{d}z\nonumber\\
\hspace{-0.3cm}&=&\hspace{-0.3cm} \psi(b)\frac{W_{q}(x)}{W_{q}(b)}+S\left(Z_{q}(x)-\frac{W_{q}(x)}{W_{q}(b)}Z_{q}(b)\right).
\end{eqnarray}
Define two auxiliary functions as follows
\begin{eqnarray}
\label{define.upsilon}
\hspace{-0.3cm}&&\hspace{-0.3cm}
\upsilon(b):=\psi(b)-SZ_{q}(b)\quad\mbox{and}\quad V(b):=\frac{W_{q}(b)}{W_{q}^{\prime}(b)}.
\end{eqnarray}

\medskip
The following result characterizes the global maximum point and the global maximum value of the value function $b\mapsto\phi(x;b)$ for fixed $x$. It should be mentioned that a verification argument similar to the one in \cite{Albrecher2007} is adopted. However, our arguments are more demanding since the inclusion of the terminal value brings in new challenges which require deeper understanding of the scale functions associated with the L\'evy process.

\medskip
\begin{thm}
\label{th.3.1}
Assume that the L\'{e}vy measure has a completely monotone density.

If $\upsilon(0)>V(0)(1-SqW_{q}(0))$, then there is a unique positive solution $b^{+}$ of the following equation
\begin{eqnarray}
\label{equ.det.argmax.}
\hspace{-0.3cm}&&\hspace{-0.3cm}
\upsilon(b)-V(b)\left(1-SqW_{q}(b)\right)=0 \Leftrightarrow
\upsilon^{\prime}(b)=1-SqW_{q}(b),
\end{eqnarray}
and the function $\phi(x;b)$ attains its global maximum at $b^{*}=b^{+}$.
Otherwise, if $\upsilon(0)\leq V(0)(1-SqW_{q}(0))$, then \eqref{equ.det.argmax.} has no positive solution, and the function $\phi(x;b)$ attains its global maximum at $b^{*}=0$.

Moreover, we have
\begin{eqnarray}
\label{}
\hspace{-0.3cm}&&\hspace{-0.3cm}
\phi(x;b^{*})=SZ_{q}(x)+\frac{W_{q}(x)}{W_{q}^{\prime}(b^{*})}\left(1-SqW_{q}(b^{*})\right).
\nonumber
\end{eqnarray}
\end{thm}

\medskip
\begin{proof}
By \eqref{phi.0}, \eqref{phi} and the L'H\^opital's rule, one has
\begin{eqnarray}
\label{1}
\hspace{-0.3cm}&&\hspace{-0.3cm}
\lim_{b \to \infty}\int_b^{\infty}\left(\frac{W_{q}(b)}{W_{q}(z)}\right)^{\frac{1}{1-\ell}}\mathrm{d}z
=\frac{1-\ell}{\Phi(q)},
\end{eqnarray}
and
\begin{eqnarray}
\label{2}
\hspace{-0.3cm}&&\hspace{-0.3cm}
\lim_{b \to \infty}\int_b^{\infty}\left(\frac{W_{q}(b)}{W_{q}(z)}\right)^{\frac{1}{1-\ell}}
\left(\frac{W_{q}^{\prime}(z)}{W_{q}(z)}Z_{q}(z)-qW_{q}(z)\right)\mathrm{d}z=0.
\end{eqnarray}
Combining \eqref{phi}, {(\ref{1})} and {(\ref{2})} yields
\begin{eqnarray}
\label{jixian0}
\hspace{-0.3cm}&&\hspace{-0.3cm}
\lim_{b \to \infty}\left[\upsilon(b)-V(b)\left(1-SqW_{q}(b)\right)\right]
\nonumber\\
\hspace{-0.3cm}&=&\hspace{-0.3cm}
\lim_{b \to \infty}\left[\frac{S}{1-\ell}\int_b^{\infty}\left(\frac{W_{q}(b)}{W_{q}(z)}\right)^{\frac{1}{1-\ell}}
\left(\frac{W_{q}^{\prime}(z)}{W_{q}(z)}Z_{q}(z)-qW_{q}(z)\right)\mathrm{d}z
+S\Big(q\frac{[W_{q}(b)]^{2}}{W_{q}^{\prime}(b)}-Z_{q}(b)\Big)\right.
\nonumber\\
\hspace{-0.3cm}&&\hspace{0.3cm}
\left.+\frac{\ell}{1-\ell}\int_b^{\infty}\left(\frac{W_{q}(b)}{W_{q}(z)}\right)^{\frac{1}{1-\ell}}\mathrm{d}z-\frac{W_{q}(b)}{W_{q}^{\prime}(b)}\right]
\nonumber\\
\hspace{-0.3cm}&=&\hspace{-0.3cm}
\frac{\ell-1}{\Phi(q)}<0.
\end{eqnarray}
By \eqref{define.upsilon}, $\phi(x;b)$ can be rewritten as
\begin{eqnarray}
\label{}
\phi(x;b)
\hspace{-0.3cm}&=&\hspace{-0.3cm}SZ_{q}(x)+\frac{W_{q}(x)}{W_{q}(b)}\,\upsilon(b),
\nonumber
\end{eqnarray}
from which one gets
\begin{eqnarray}
\label{daoshu}
\frac{\partial\phi(x;b)}{\partial b}
\hspace{-0.3cm}&=&\hspace{-0.3cm}\frac{W_{q}(x)}{[W_{q}(b)]^{2}}\left[W_{q}(b)\upsilon^{\prime}(b)-W_{q}^{\prime}(b)
\,\upsilon(b)\right]
\nonumber\\
\hspace{-0.3cm}&=&\hspace{-0.3cm}\frac{\ell}{1-\ell}\frac{W_{q}(x)W_{q}^{\prime}(b)}{[W_{q}(b)]^{2}}
\left[\upsilon(b)-V(b)\left(1-SZ_{q}^{\prime}(b)\right)\right],
\end{eqnarray}
where the following expression for $\upsilon^{\prime}(b)$ is used to derive the second equality in \eqref{daoshu}
\begin{eqnarray}
\label{}
\upsilon^{\prime}(b)
\hspace{-0.3cm}&=&\hspace{-0.3cm}
\frac{1}{1-\ell}\frac{\upsilon(b)}{V(b)}+\frac{\ell}{1-\ell}\left(SZ_{q}^{\prime}(b)-1\right).
\nonumber
\end{eqnarray}
According to {(\ref{daoshu})}, $\frac{\partial\phi(x;b)}{\partial b}=0$ is equivalent to
\eqref{equ.det.argmax.},
i.e., $b_{0}$ is a critical point of $b\mapsto\phi(x;b)$ if $b_{0}$ solves \eqref{equ.det.argmax.}.
To specify the nature of this critical point, we examine the second partial derivative of $\phi(x; b)$ with respect to $b$,
\begin{eqnarray}
\label{erjiedao}
\left.\frac{\partial^{2}\phi(x;b)}{\partial b^{2}}\right|_{b=b_{0}}
\hspace{-0.3cm}&=&\hspace{-0.3cm}
\frac{\ell}{1-\ell}\frac{W_{q}(x){W_{q}^{\prime\prime}(b_{0})\upsilon(b_{0})}}{[W_{q}(b_{0})]^{2}}
+\frac{\ell}{1-\ell}\frac{SqW_{q}(x)W_{q}^{\prime}(b_{0})}{W_{q}(b_{0})}
\nonumber\\
\hspace{-0.3cm}&=&\hspace{-0.3cm}
\frac{\ell}{1-\ell}\frac{W_{q}(x)}{W_{q}(b_{0})}\left(\frac{1-SqW_{q}(b_{0})}{W_{q}^{\prime}(b_{0})}W_{q}^{\prime\prime}(b_{0})+SqW_{q}^{\prime}(b_{0})\right)\nonumber\\
\hspace{-0.3cm}&=&\hspace{-0.3cm}\frac{\ell}{1-\ell}\frac{W_{q}(x)}{W_{q}(b_{0})}f(b_0),
\end{eqnarray}
where
\begin{equation}\label{f}
f(b):=\frac{1-SqW_{q}(b)}{W_{q}^{\prime}(b)}W_{q}^{\prime\prime}(b)+SqW_{q}^{\prime}(b).
\end{equation}
In what follows, we are devoted to ruling out the possibility that the map $b\mapsto\phi(x;b)$ has local minimum points or saddle points in $[0,\infty)$, by separately considering two opposite cases, i.e., {\bf Case A}: $S\leq0$ and {\bf Case B}: $S>0$.

\medskip
\noindent {\bf Case A}. $S\leq0$. In this case, we have $1-SqW_{q}(b)>0$ for all $b\in[0,\infty)$, and the following two claims {\bf A1} and {\bf A2} hold true.

\begin{enumerate}
\item[\bf{A1.}] {\it The map $b\mapsto\phi(x;b)$ cannot have a local minimum point in $[0,\infty)$.}\medskip

Indeed, by {(\ref{jixian0})} and {(\ref{daoshu})}, there must exist a sufficiently large $\overline{b}>0$ such that $\phi(x;b)$ is decreasing in $b$ over $[\overline{b},\infty)$. Hence, if $b_{1}$ is a local minimum point of $b\mapsto\phi(x;b)$, then there would have to exist a local maximum  $b_{2}\in(b_{1},\overline{b}\,]$, i.e.
\begin{eqnarray}
\label{maximum}
\hspace{-0.3cm}&&\hspace{-0.3cm}
\left.\frac{\partial^{2}\phi(x;b)}{\partial b^{2}}\right|_{b=b_{1}}>0
\quad\mathrm{and} \quad \left.\frac{\partial^{2}\phi(x;b)}{\partial b^{2}}\right|_{b=b_{2}}<0,\nonumber
\end{eqnarray}
or equivalently,
 \begin{eqnarray}
\label{maximum1}
\hspace{-0.3cm}&&\hspace{-0.3cm}
f(b_{1})>0
\quad\mathrm{and} \quad f(b_{2})<0.
\end{eqnarray}
However, as we assumed that the L\'evy measure of $X$ has a completely monotone density, by Corollary 1 of \cite{Loeffen2009}, we know
\begin{eqnarray}
\label{sanpie}
\hspace{-0.3cm}&&\hspace{-0.3cm}
W_{q}^{\prime}(b)W_{q}^{\prime\prime\prime}(b)-\left(W_{q}^{\prime\prime}(b)\right)^{2}>0,\quad b\in(0,\infty),
\end{eqnarray}
which gives
\begin{eqnarray}
\label{imp.fac.}
\hspace{-0.3cm}&&\hspace{-0.3cm}
f^{\prime}(b)=
\left(1-SqW_{q}(b)\right)\frac{W_{q}^{\prime}(b)W_{q}^{\prime\prime\prime}(b)
-[W_{q}^{\prime\prime}(b)]^{2}}{[W_{q}^{\prime}(b)]^{2}}
>0,\quad b\in(0,\infty).
\end{eqnarray}
Clearly, \eqref{imp.fac.} contradicts to {(\ref{maximum1})}, therefore claim {\bf A1} is true.

\item[\bf{A2}.] {\it The map $b\mapsto\phi(x;b)$ cannot have a saddle point $b_{0}$ in $[0,\infty)$.}\medskip

Since
\begin{eqnarray}
\label{}
\hspace{-0.3cm}&&\hspace{-0.3cm}
\upsilon(b_{0})=V(b_{0})\left(1-SqW_{q}(b_{0})\right) \quad\mathrm{and} \quad
\left.\frac{\partial^{2}\phi(x;b)}{\partial b^{2}}\right|_{b=b_{0}}=0,
\nonumber
\end{eqnarray}
one can verify that
\begin{eqnarray}
\label{}
\hspace{-0.3cm}&&\hspace{-0.3cm}
W_{q}^{\prime\prime}(b_{0})=-\frac{Sq[W_{q}^{\prime}(b_{0})]^{2}}{1-SqW_{q}(b_{0})}
\quad\mathrm{and} \quad
V^{\prime}(b_{0})=\frac{1}{1-SqW_{q}(b_{0})}.
\nonumber
\end{eqnarray}
For convenience, let $$h(b):=\upsilon(b)-V(b)\left(1-SqW_{q}(b)\right).$$
Clearly, $h(b_0)=0$. We will show that $b_{0}$ is a local minimum point of $h$. The first and second derivatives of $h$ are
\begin{eqnarray}
\label{}
h^{\prime}(b)
\hspace{-0.3cm}&=&\hspace{-0.3cm}
\upsilon^{\prime}(b)-V^{\prime}(b)\left(1-SqW_{q}(b)\right)+SqW_{q}^{\prime}(b)V(b)
\nonumber\\
\hspace{-0.3cm}&=&\hspace{-0.3cm}
\frac{\upsilon(b)}{(1-\ell)V(b)}+\frac{\ell}{1-\ell}\left(SqW_{q}(b)-1\right)
-\Big(1-\frac{W_{q}(b)W_{q}^{\prime\prime}(b)}{[W_{q}^{\prime}(b)]^{2}}\Big)\left(1-SqW_{q}(b)\right)+SqW_{q}(b)
\nonumber\\
\hspace{-0.3cm}&=&\hspace{-0.3cm}
\frac{h(b)}{(1-\ell)}+SqW_{q}(b)
+\frac{W_{q}(b)W_{q}^{\prime\prime}(b)}{[W_{q}^{\prime}(b)]^{2}}\left(1-SqW_{q}(b)\right),
\nonumber
\end{eqnarray}
and
\begin{eqnarray}
\label{}
h^{\prime\prime}(b)
\hspace{-0.3cm}&=&\hspace{-0.3cm}
\frac{h'(b)}{(1-\ell)}+SqW_{q}^{\prime}(b)-\frac{W_{q}(b)W_{q}^{\prime\prime}(b)}{[W_{q}^{\prime}(b)]^{2}}SqW_{q}^{\prime}(b)
\nonumber\\
\hspace{-0.3cm}&&\hspace{-0.3cm}
+\frac{\left(W_{q}^{\prime}(b)W_{q}^{\prime\prime}(b)+W_{q}(b)W_{q}^{\prime\prime\prime}(b)\right)W_{q}^{\prime}(b)-
2W_{q}(b)[W_{q}^{\prime\prime}(b)]^{2}}
{[W_{q}^{\prime}(b)]^3}\left(1-SqW_{q}(b)\right)
\nonumber\\
\hspace{-0.3cm}&=&\hspace{-0.3cm}
\frac{h'(b)}{(1-\ell)}+SqW_{q}^{\prime}(b)
+\frac{W_{q}^{\prime\prime}(b)}{W_{q}^{\prime}(b)}\left(1-SqW_{q}(b)\right)
+\frac{W_{q}(b)W_{q}^{\prime\prime\prime}(b)}{[W_{q}^{\prime}(b)]^{2}}\left(1-SqW_{q}(b)\right)\nonumber\\
\hspace{-0.3cm}&&\hspace{-0.3cm}
-\frac{2W_{q}(b)[W_{q}^{\prime\prime}(b)]^{2}}{[W_{q}^{\prime}(b)]^{3}}\left(1-SqW_{q}(b)\right)
-\frac{SqW_{q}(b)W_{q}^{\prime\prime}(b)}{W_{q}^{\prime}(b)}.
\nonumber
\end{eqnarray}
One can see that $h^{\prime}(b_{0})=0$, and
\begin{eqnarray*}
\hspace{-0.3cm}&&\hspace{-0.3cm}
SqW_{q}^{\prime}(b_0)
+\frac{W_{q}^{\prime\prime}(b_0)}{W_{q}^{\prime}(b_0)}\left(1-SqW_{q}(b_0)\right)=0.
\end{eqnarray*}
Also, we have
\begin{eqnarray}
-\frac{2W_{q}(b_{0})[W_{q}^{\prime\prime}(b_{0})]^{2}}{[W_{q}^{\prime}(b_{0})]^{3}}\left(1-SqW_{q}(b_{0})\right)
-\frac{SqW_{q}(b_{0})W_{q}^{\prime\prime}(b_{0})}{W_{q}^{\prime}(b_{0})}
\hspace{-0.3cm}&=&\hspace{-0.3cm}
-\frac{S^{2}q^{2}W_{q}(b_{0})W_{q}^{\prime}(b_{0})}{1-SqW_{q}(b_{0})}.
\nonumber
\end{eqnarray}
By \eqref{sanpie} we obtain
\begin{eqnarray}
\label{}
h^{\prime\prime}(b_{0})
\hspace{-0.3cm}&=&\hspace{-0.3cm}
\frac{W_{q}(b_{0})W_{q}^{\prime\prime\prime}(b_{0})}{[W_{q}^{\prime}(b_{0})]^{2}}\left(1-SqW_{q}(b_{0})\right)
-\frac{S^{2}q^{2}W_{q}(b_{0})W_{q}^{\prime}(b_{0})}{1-SqW_{q}(b_{0})}
\nonumber\\
\hspace{-0.3cm}&>&\hspace{-0.3cm}
\frac{W_{q}(b_{0})[W_{q}^{''}(b_{0})]^2}{[W_{q}^{\prime}(b_{0})]^3}\left(1-SqW_{q}(b_{0})\right)
-\frac{S^{2}q^{2}W_{q}(b_{0})W_{q}^{\prime}(b_{0})}{1-SqW_{q}(b_{0})}\nonumber\\
\hspace{-0.3cm}&=&\hspace{-0.3cm}
\frac{W_{q}(b_{0})\left(1-SqW_{q}(b_{0})\right)}{[W_{q}^{\prime}(b_{0})]^3}
\Big[(W_{q}^{''}(b_{0}))^2-\frac{S^{2}q^{2}(W_{q}^{\prime}(b_{0}))^4}{(1-SqW_{q}(b_{0}))^2}\Big]=0.
\nonumber
\end{eqnarray}
As a result, $h(b)$ reaches a local minimum value of $0$ at $b_{0}$. Hence, $h(b)>0$ holds true for all $b\in(b_{0},b_{0}+\epsilon_{0})$ for some $\epsilon_{0}>0$, which combined with {(\ref{daoshu})} implies $\frac{\partial\phi(x;b)}{\partial b}>0$ for all $b\in(b_{0},b_{0}+\epsilon_{0})$. Therefore, the saddle point $b_{0}$ of the map $b\mapsto\phi(x;b)$, has to be followed by a maximum or another saddle point $b_{1}\in(b_{0},\infty)$ of the map  $b\mapsto\phi(x;b)$, i.e., $f(b_{0})=0$ and $f(b_{1})\leq 0$ (c.f., {(\ref{erjiedao})} and {(\ref{f})}), contradicting to \eqref{imp.fac.}. Hence the claim {\bf A2} is true.
\end{enumerate}

\medskip
{\bf Case B}. $S>0$. The following two claims {\bf B1} and {\bf B2} are true.\medskip

\begin{enumerate}
\item[{\bf B1.}] {\it The map $b\mapsto\phi(x;b)$ cannot have a local minimum in $[0,\infty)$.}\medskip

If there exists a local minimum point $b_{1}$ of the map $b\mapsto\phi(x;b)$, using {(\ref{jixian0})} and {(\ref{daoshu})} once again we can see that there exists a local maximum point $b_{2}\in(b_{1},\infty)$, and then we obtain (c.f., {(\ref{erjiedao})} and {(\ref{f})})
\begin{eqnarray}
\label{fbudeng}
\hspace{-0.3cm}&&\hspace{-0.3cm}
f(b_{1})>0
\quad\mathrm{and} \quad f(b_{2})<0.
\end{eqnarray}
Combining {(\ref{f})} with \eqref{W} yields
\begin{eqnarray}
\label{maximumf(b)}
f(b)\hspace{-0.3cm}&=&\hspace{-0.3cm}
\frac{1-SqW_{q}(b)}{W_{q}^{\prime}(b)}W_{q}^{\prime\prime}(b)+SqW_{q}^{\prime}(b)\nonumber\\
\hspace{-0.3cm}&=&\hspace{-0.3cm}
\frac{W_{q}^{\prime\prime}(b)}{W_{q}^{\prime}(b)}-\frac{SqW_{q}(b)W_{q}^{\prime\prime}(b)}{W_{q}^{\prime}(b)}+SqW_{q}^{\prime}(b)
>
\frac{W_{q}^{\prime\prime}(b)}{W_{q}^{\prime}(b)},\nonumber
\end{eqnarray}
which, together with \eqref{fbudeng}, implies that
\begin{eqnarray}
\label{maximumf(b)}
W_{q}^{\prime\prime}(b_{2})<0.
\end{eqnarray}
\begin{itemize}
    \item If $1-SqW_{q}(b_{2})\geq0$, the strict increasing property of the scale function $W_q(b)$ gives $1-SqW_{q}(b)\geq0$ for $b \in [b_{1},b_{2}]$. Similar to the derivation of result \eqref{imp.fac.}, we can show that $f'(b)>0$ for $b \in [b_{1},b_{2}]$, which contradicts to the result of $f(b_2)<0$ given in \eqref{fbudeng}.\smallskip
    \item Otherwise, if $1-SqW_{q}(b_{2})<0$, from $f(b_{2})=\frac{1-SqW_{q}(b_{2})}{W_{q}^{\prime}(b_{2})}W_{q}^{\prime\prime}(b_{2})+SqW_{q}^{\prime}(b_{2})<0$, it holds that $W_{q}^{\prime\prime}(b_{2})>\frac{-Sq[W_{q}^{\prime}(b_{2})]^{2}}{1-SqW_{q}(b_{2})}>0$, which contradicts to \eqref{maximumf(b)}.
\end{itemize}\medskip
Therefore, the claim {\bf B1} is true.\medskip

\item[{\bf B2.}] {\it The map $b\mapsto\phi(x;b)$ cannot have a saddle point $b_{0}$ in $[0,\infty)$.}\medskip

Assume that $b_0$ is a saddle point in $[0,\infty)$.\smallskip
\begin{itemize}
    \item If $1-SqW_{q}(b_{0})>0$, adopting the same arguments as in the claim {\bf A$_{2}$}, we can show that $h(b)$ reaches a local minimum value of $0$ at $b_{0}$, which, together with {(\ref{jixian0})} and {(\ref{daoshu})}, implies that there must be a local maximum point or another saddle point $b_{1}\in(b_{0},\infty)$ of the map $b\mapsto\phi(x;b)$, i.e.
\begin{eqnarray}
\label{maximumf(b00)}
f(b_{1})\leq0=f(b_{0}).
\end{eqnarray}
If $1-SqW_{q}(b_{1})\geq0$, \eqref{imp.fac.} holds true for all $b \in [b_{0},b_{1})$, contradicting to \eqref{maximumf(b00)}. Otherwise, if $1-SqW_{q}(b_{1})<0$, by $
\frac{1-SqW_{q}(b_{1})}{W_{q}^{\prime}(b_{1})}W_{q}^{\prime\prime}(b_{1})+SqW_{q}^{\prime}(b_{1})
=f(b_{1})\leq0$ we get $W_{q}^{\prime\prime}(b_{1})\geq\frac{-Sq[W_{q}^{\prime}(b_{1})]^{2}}{1-SqW_{q}(b_{1})}>0$, contradicting to $W_{q}^{\prime\prime}(b_{1})\leq0$ (c.f., arguments used in verifying \eqref{maximumf(b)}).\medskip

\item If $1-SqW_{q}(b_{0})\leq0$, adopting the same arguments as in verifying {(\ref{maximumf(b)})} we have $W_{q}^{\prime\prime}(b_{0})<0$. However, by $f(b_{0})=0$, one can see that $W_{q}^{\prime\prime}(b_{0})=-\frac{Sq[W_{q}^{\prime}(b_{0})]^{2}}{1-SqW_{q}(b_{0})}>0$, which is a contradiction.
\end{itemize}\medskip
Hence, the claim {\bf B2} is true.
\end{enumerate}

Summing up the arguments concerning {\bf Case A} and {\bf Case B}, we conclude that \eqref{equ.det.argmax.} has at most one positive solution.\medskip
\begin{itemize}
    \item When $\upsilon(0)>V(0)(1-SqW_{q}(0))$, by {(\ref{jixian0})}, one concludes that a positive solution of \eqref{equ.det.argmax.}, denoted by $b^{+}\in(0,\infty)$, exists and satisfies \begin{eqnarray}
    \label{maximumf(b00)01}
    \hspace{1cm}\upsilon(b)>V(b)(1-SqW_{q}(b)),\,\,b\in[0,b^{+})\quad \mbox{and} \quad \upsilon(b)<V(b)(1-SqW_{q}(b)),\,\,b\in(b^{+},\infty),\nonumber
    \end{eqnarray}
    i.e., the map $b\mapsto\phi(x;b)$ reaches its global maximum value at $b^{*}=b^{+}$.\medskip

\item When $\upsilon(0)<V(0)(1-SqW_{q}(0))$, the equation \eqref{equ.det.argmax.} has no positive solution and
\begin{eqnarray}
\label{maximumf(b00)002}
\upsilon(b)<V(b)(1-SqW_{q}(b)),\,\,b\in[0,\infty),
\end{eqnarray}
i.e., the map $b\mapsto\phi(x;b)$ reaches its global maximum value at $b^{*}=0$.\medskip
\item When $\upsilon(0)=V(0)(1-SqW_{q}(0))$, then \eqref{maximumf(b00)002} holds true for $b\in(0,\infty)$ and the map $b\mapsto\phi(x;b)$ reaches its global maximum value at $b^{*}=0$.
\end{itemize}\medskip
The proof is completed.
\end{proof}

\begin{rem}
The assumption that the L\'evy measure of $X$ has a completely monotone density was firstly introduced in \cite{Loeffen2009} to identify a sufficient condition, under which a barrier dividend strategy is the optimal dividend strategy producing the largest expected total discounted dividend payments.
Also, with this assumption, the scale functions have some desirable analytic properties.
\end{rem}

\begin{rem} When $S=0$, we can see that
\begin{eqnarray}
\label{}
\hspace{-0.3cm}&&\hspace{-0.3cm}
\phi(x;b)=\frac{W_{q}(x)}{W_{q}(b)}\psi(b),\qquad \upsilon(b)=\psi(b)=\frac{\ell}{1-\ell}\int_b^{\infty}\left(\frac{W_{q}(b)}{W_{q}(z)}\right)^{\frac{1}{1-\ell}}\mathrm{d}z,
\nonumber
\end{eqnarray}
and $\frac{\partial\phi(x;b)}{\partial b}=0$ is equivalent to
\begin{eqnarray}
\label{}
\hspace{-0.3cm}&&\hspace{-0.3cm}
\upsilon(b)=V(b) \quad \mathrm{or} \quad \upsilon^{\prime}(b)=1,
\nonumber
\end{eqnarray}
which has at most one positive solution.
If $\upsilon(0)\leq V(0)$, then the map $b\mapsto\phi(x;b)$ attains its global maximum at $b^{*}=0$. If $\upsilon(0)> V(0)$,  the map $b\mapsto\phi(x;b)$ attains its global maximum at the positive solution $b^{*}$ of \eqref{equ.det.argmax.}. Thus, the global maximum value of $\phi(x;b)$ is given by
\begin{eqnarray}
\label{}
\hspace{-0.3cm}&&\hspace{-0.3cm}
\phi(x;b^{*})=\frac{W_{q}(x)}{W_{q}(b^{*})}\upsilon(b^{*})=\frac{W_{q}(x)}{W_{q}^{\prime}(b^{*})}.
\nonumber
\end{eqnarray}
These results coincide with the ones obtained in \cite{Albrecher2007}.
\end{rem}

\section{Optimal delay of taxation implementation with penalized capital injection}

\medskip
This section is devoted to studying the optimal implementation delay of taxation for a L\'evy risk process with capital injections that reflect the risk process at its infimum (or, at $0$). Taking capital injections into consideration under the L\'evy risk process, one can see that the surplus process never goes bankrupt, and hence taxation payments are expected to be collectible up to $\infty$ rather than up to the time of ruin as in the case without capital injections. No terminal value will be considered in this section.\smallskip

In \cite{Albrecher2014b}, a risk process refracted from both above and below with rates less than 1 was defined path by path through a recursive algorithm. The risk process $\{R(t);t\geq0\}$ with taxation and capital injections, which serves as the main object of this section as well as in \cite{Albrecher2014b}, is actually a risk process that is refracted from above with rate $\ell<1$ and refracted from below with rate 1. In order to better develop our arguments, we re-define the risk process $\{R(t);t\geq0\}$ in a more transparent way as follows.

\begin{itemize}
    \item Let 
$\sigma_{0,1}^{-}=\sigma_{0}^{-}$ 
and $U_{1}(t)=U(t)$ for $t\geq0$. If $\sigma_{0,1}^{-}<\infty$, then define
$$R_{1}(t)=X(t)-X(\sigma_{0,1}^{-})
-\inf_{s\in[\sigma_{0,1}^{-},t]}
\left(X(t)-X(\sigma_{0,1}^{-})\right)\wedge 0,\quad t\in[\sigma_{0,1}^{-},\infty),$$
and
$$\mathrm{H}_{1}=\sup_{0\leq s< \sigma_{0,1}^{-}}U_{1}(s),\quad \sigma_{\mathrm{H}_{1},1}^{+}=\inf\{t\geq\sigma_{0,1}^{-}; R_{1}(t)>\mathrm{H}_{1}\}.$$
Define the risk process $R$ over the time interval $[0,\sigma_{\mathrm{H}_{1},1}^{+})$ as
$$
\begin{array}{rcl}
R(t)\hspace{-0.3cm}&=&\hspace{-0.3cm}U_{1}(t),\quad \mbox{if}\quad 0\leq t<\sigma_{0,1}^{-},\\
R(t)\hspace{-0.3cm}&=&\hspace{-0.3cm}R_{1}(t),\quad \mbox{if}\quad \sigma_{0,1}^{-}\leq t<\sigma_{\mathrm{H}_{1},1}^{+}\mbox{ and }\,\sigma_{0,1}^{-}<\infty.
\end{array}
$$

\item For $n\geq1$, if $\sigma_{\mathrm{H}_{n},n}^{+}<\infty$ (or, equivalently, $\sigma_{0,n}^{-}<\infty$), let
\begin{eqnarray}
U_{n+1}(t)
\hspace{-0.3cm}&=&\hspace{-0.3cm}\mathrm{H}_{n}
+X(t)-X(\sigma_{\mathrm{H}_{n},n}^{+})
\nonumber\\
\hspace{-0.3cm}&&\hspace{-0.3cm}
-\ell \Big(\sup_{\sigma_{\mathrm{H}_{n},n}^{+}\leq s\leq t} X(s)-X(\sigma_{\mathrm{H}_{n},n}^{+})\Big),\quad t\in[\sigma_{\mathrm{H}_{n},n}^{+},\infty),\nonumber
\end{eqnarray}
and
$$\sigma_{0,n+1}^{-}=\inf\{t\geq \sigma_{\mathrm{H}_{n},n}^{+}; U_{n+1}(t)<0\},\quad \mathrm{H}_{n+1}=\sup_{\sigma_{\mathrm{H}_{n},n}^{+}\leq s< \sigma_{0,n+1}^{-}}U_{n+1}(s),$$
with the convention of $\inf\emptyset=\infty$.
If $\sigma_{0,n+1}^{-}<\infty$, let further
$$R_{n+1}(t)=X(t)-X(\sigma_{0,n+1}^{-})
-\inf_{s\in[\sigma_{0,n+1}^{-},t]}
\left(X(t)-X(\sigma_{0,n+1}^{-})\right)\wedge 0,\quad t\in[\sigma_{0,n+1}^{-},\infty),$$
and
$$\sigma_{\mathrm{H}_{n+1},n+1}^{+}=\inf\{t\geq\sigma_{0,n+1}^{-}; R_{n+1}(t)>\mathrm{H}_{n+1}\}.$$
Accordingly, 
define the risk process $R$ over the time interval $[\sigma_{\mathrm{H}_{n},n}^{+},\sigma_{\mathrm{H}_{n+1},n+1}^{+})$ as
$$
\begin{array}{rcl}
R(t)\hspace{-0.3cm}&=&\hspace{-0.3cm}U_{n+1}(t),\quad \mbox{if}\quad \sigma_{\mathrm{H}_{n},n}^{+}\leq t<\sigma_{0,n+1}^{-}\mbox{ and }\,\sigma_{\mathrm{H}_{n},n}^{+}<\infty,\\
R(t)\hspace{-0.3cm}&=&\hspace{-0.3cm}R_{n+1}(t),\quad \mbox{if}\quad \sigma_{0,n+1}^{-}\leq t<\sigma_{\mathrm{H}_{n+1},n+1}^{+}\mbox{ and }\,\sigma_{0,n+1}^{-}<\infty.
\end{array}
$$
\end{itemize}
It can be verified that $\lim\limits_{n\rightarrow \infty}\sigma_{0,n}^{-}=\lim\limits_{n\rightarrow \infty}\sigma_{\mathrm{H}_{n},n}^{+}=\infty$ almost surely under $\mathbb{P}_{x}$ with $x\in(0,\infty)$. Hence the process $R$ constructed as above is well defined.

For $a\in(0,\infty)$, let
$\kappa_{a}^{+}:=\inf\{t\geq0; R(t)>a\}$
be the first up-crossing time of level $a$ for the risk process $R$.
In addition, for $a\in(0,\infty),\,x\in(0,a]$, we define
$$f_{a}(x):=
\mathbb{E}_{x}\left[\mathrm{e}^{-q \kappa_{a}^{+}}\right],$$
and
\begin{eqnarray}
\hspace{-0.1cm}g_{a}(x)
\hspace{-0.3cm}&:=&\hspace{-0.3cm}
\mathbb{E}_{x}\bigg[\sum_{n=0}^{\infty}\ell\bigg(\int_{\sigma_{\mathrm{H}_{n},n}^{+}}^{\sigma_{0,n+1}^{-}}\mathrm{e}^{-q t} \mathrm{d} \Big(\sup_{\sigma_{\mathrm{H}_{n},n}^{+}\leq s\leq t} X(s)-X(\sigma_{\mathrm{H}_{n},n}^{+})\Big)
\mathbf{1}_{\{\sigma_{0,n+1}^{-}\leq \kappa_{a}^{+}\}}
\nonumber\\
\hspace{-0.3cm}&&\hspace{-1cm}
+
\int_{\sigma_{\mathrm{H}_{n},n}^{+}}^{\kappa_{a}^{+}}\mathrm{e}^{-q t}\mathrm{d} \Big(\sup_{\sigma_{\mathrm{H}_{n},n}^{+}\leq s\leq t} X(s)-X(\sigma_{\mathrm{H}_{n},n}^{+})\Big)\mathbf{1}_{\{\sigma_{\mathrm{H}_{n},n}^{+}\leq \kappa_{a}^{+}<\sigma_{0,n+1}^{-}\}}
\bigg)\bigg],
\nonumber
\end{eqnarray}
where
$\sigma_{\mathrm{H}_{0},0}^{+}:=0$, and
\begin{eqnarray}
r_{a}(x)\hspace{-0.3cm}&:=&\hspace{-0.3cm}
\mathbb{E}_{x}\bigg[\sum_{n=1}^{\infty}
\bigg(\mathrm{e}^{-q \sigma_{0,n}^{-}}
\left|U_{n}(\sigma_{0,n}^{-})\right|
-
\int_{\sigma_{0,n}^{-}}
^{\sigma_{\mathrm{H}_{n},n}^{+}}
\mathrm{e}^{-q t}\mathrm{d}\Big(\inf_{s\in[\sigma_{0,n}^{-},t]}
\left(X(t)-X(\sigma_{0,n}^{-})\right)\wedge 0\Big)
\bigg)
\nonumber\\
\hspace{-0.3cm}&&\hspace{-0.3cm}
\times\mathbf{1}_{\{\sigma_{\mathrm{H}_{n},n}^{+}\leq \kappa_{a}^{+}\}}\bigg].
\nonumber
\end{eqnarray}
Hence, $g_{a}(x)$ represents the expected total discounted tax payments collected until the risk process $R$ reaches $a$, and $r_{a}(x)$ represents the expected total discounted capital injections made until the risk process $R$ reaches $a$.

\medskip
The following result expresses the functions $f_{a}$, $ g_{a}$ and $ r_{a}$ by the scale functions $W_{q}$ and $Z_{q}$.
It should be mentioned that the two identities in \eqref{two.sid.tax.inj.} can also be found in \cite{Albrecher2014b}. However, a new and transparent argument involving the excursion theory associated with L\'evy processes is adopted in our case.

\medskip
\begin{prop}
\label{pro4.1}
For $q\geq0$, $a\in(0,\infty)$, and $x\in (0, a]$, we have
\begin{eqnarray}
\label{two.sid.tax.inj.}
f_{a}(x)=\left(\frac{Z_{q}(x)}{Z_{q}(a)}\right)^{\frac{1}{1-\ell}},\quad
g_{a}(x)=
\frac{\ell}{1-\ell}\int_x^{a}\left(\frac{Z_{q}(x)}{Z_{q}(w)}\right)^{\frac{1}{1-\ell}}\mathrm{d}w,
\end{eqnarray}
and
\begin{eqnarray}
\label{e.d.a.t.01}
\quad\quad \,\,\,r_{a}(x)=
\frac{1}{1-\ell}\int_x^{a}
\left(Z_{q}(w)-\left(\overline{Z}_{q}(w)+\frac{\Psi^{\prime}(0+)}{q}\right)
\frac{qW_{q}(w)}{Z_{q}(w)}\right)
\left(\frac{Z_{q}(x)}{Z_{q}(w)}\right)^{\frac{1}{1-\ell}}\mathrm{d}w.
\end{eqnarray}
\end{prop}

\begin{proof}
From \eqref{exi.pro.}, \eqref{12} and the Markov property, it is seen that, for $a\in(0,\infty)$, $x\in(0,a]$,
\begin{eqnarray}
\label{gen.tax}
f_{a}(x)\hspace{-0.3cm}&=&\hspace{-0.3cm}
\mathbb{E}_{x}\left[\mathrm{e}^{-q \sigma_{a}^{+}};\sigma_{a}^{+}<\sigma_{0}^{-}\right]
+\mathbb{E}_{x}\left[\mathrm{e}^{-q \sigma_{0}^{-}}\mathbb{E}_{0}\Big[\mathrm{e}^{-q \rho_{\mathrm{H}_{1}}^{+}}\Big]
\mathbb{E}_{\mathrm{H}_{1}}
\left[\mathrm{e}^{-q \kappa_{a}^{+}}\right];\sigma_{0}^{-}<\sigma_{a}^{+}\right]
\nonumber\\
\hspace{-0.3cm}&=&\hspace{-0.3cm}
\left(\frac{W_{q}(x)}
{W_{q}(a)}\right)^{\frac{1}{1-\ell}}
+\mathbb{E}_{x}\left[\mathrm{e}^{-q \sigma_{0}^{-}}
\frac{f_{a}(\mathrm{H}_{1})}{Z_{q}(\mathrm{H}_{1})}
;\sigma_{0}^{-}<\sigma_{a}^{+}\right]
\nonumber\\
\hspace{-0.3cm}&=&\hspace{-0.3cm}
\left(\frac{W_{q}(x)}
{W_{q}(a)}\right)^{\frac{1}{1-\ell}}
+
\frac{1}{1-\ell}\int_{x}^{a}\frac{f_{a}\left(s\right)}{Z_{q}(s)}
\left(\frac{W_{q}(x)}
{W_{q}(s)}\right)^{\frac{1}{1-\ell}}
\left(\frac{W_{q}^{\prime}(s)}{W_{q}(s)}Z_{q}(s)
-qW_{q}(s)\right)
\mathrm{d}s.
\end{eqnarray}
Differentiating both sides of \eqref{gen.tax} with respect to $x$ gives the following differential equation, for $x\in (0, \infty)$,
\begin{eqnarray}
\label{dif.eq.}
f_{a}^{\prime}(x)\hspace{-0.3cm}&=&\hspace{-0.3cm}
\frac{1}{1-\ell}
\frac{W_{q}^{\prime}(x)}
{W_{q}(x)}
\left(\frac{W_{q}(x)}
{W_{q}(a)}\right)^{\frac{1}{1-\ell}}
+\frac{1}{1-\ell}\frac{W_{q}^{\prime}(x)}
{W_{q}(x)}\left(f_{a}(x)-
\left(\frac{W_{q}(x)}
{W_{q}(a)}\right)^{\frac{1}{1-\ell}}\right)\nonumber\\
\hspace{-0.3cm}&&\hspace{-0.3cm}
-\frac{f_{a}\left(x\right)}{1-\ell}
\left(\frac{W_{q}^{\prime}(x)}{W_{q}(x)}-\frac{qW_{q}(x)}{Z_{q}(x)}\right)
\nonumber\\
\hspace{-0.3cm}&=&\hspace{-0.3cm}
\frac{1}{1-\ell}
\frac{qW_{q}(x)}{Z_{q}(x)}f_{a}\left(x\right)
,
\end{eqnarray}
with boundary condition $f_{a}(a)=1$. Solving \eqref{dif.eq.}, we obtain the first identity in \eqref{two.sid.tax.inj.}.
By \eqref{e.d.a.t.} and an argument similar to the one used in deriving \eqref{gen.tax}, we have
\begin{eqnarray}
\label{gen.tax.1}
g_{a}(x)
\hspace{-0.3cm}&=&\hspace{-0.3cm}
\mathbb{E}_{x}\left[\int_{0}^{\sigma_{a}^{+}\wedge\sigma_{0}^{-}}\mathrm{e}^{-q t}\ell\,\mathrm{d}\overline{X}(t)\right]
+\mathbb{E}_{x}\left[\mathrm{e}^{-q \sigma_{0}^{-}}
\frac{g(\mathrm{H}_{1})}{Z_{q}(\mathrm{H}_{1})}
;\sigma_{0}^{-}<\sigma_{a}^{+}\right]
\nonumber\\
\hspace{-0.3cm}&=&\hspace{-0.3cm}
\frac{\ell}{1-\ell}\int_{x}^{a}
\left(\frac{W_{q}(x)}
{W_{q}(s)}\right)^{\frac{1}{1-\ell}}
\mathrm{d}s
+
\frac{1}{1-\ell}\int_{x}^{a}\frac{g\left(s\right)}{Z_{q}(s)}
\left(\frac{W_{q}(x)}
{W_{q}(s)}\right)^{\frac{1}{1-\ell}}
\nonumber\\
&&\quad\times
\left(\frac{W_{q}^{\prime}(s)}{W_{q}(s)}Z_{q}(s)
-qW_{q}(s)\right)
\mathrm{d}s
,\qquad x\in(0,\infty),\quad a\in(x,\infty).
\end{eqnarray}
Differentiating both sides of \eqref{gen.tax.1} with respect to $x$ gives the following differential equation
\begin{eqnarray}
\label{dif.eq.1}
g_{a}^{\prime}(x)\hspace{-0.3cm}&=&\hspace{-0.3cm}
\frac{1}{1-\ell}\frac{W_{q}^{\prime}(x)}
{W_{q}(x)}g(x)
-\frac{\ell}{1-\ell}
-
\frac{g\left(x\right)}{1-\ell}
\left(\frac{W_{q}^{\prime}(x)}{W_{q}(x)}-\frac{qW_{q}(x)}{Z_{q}(x)}\right)
\nonumber\\
\hspace{-0.3cm}&=&\hspace{-0.3cm}
\frac{1}{1-\ell}
\frac{qW_{q}(x)}{Z_{q}(x)}g\left(x\right)-\frac{\ell}{1-\ell}
,\qquad x\in(0,\infty),
\end{eqnarray}
with boundary condition $g_{a}(a)=0$. Solving \eqref{dif.eq.1} gives the second identity in \eqref{two.sid.tax.inj.}.\smallskip

Combining \eqref{exp.cap.inj.unt.exit.},
\eqref{12} and \eqref{14} yields, for $a\in(0,\infty)$,
\begin{eqnarray}
\label{capt.inj.}
r_{a}(x)
\hspace{-0.3cm}&=&\hspace{-0.3cm}
\mathbb{E}_{x}\left[\mathrm{e}^{-q \sigma_{0}^{-}}\left|U(\sigma_{0}^{-})\right|;\sigma_{0}^{-}<\sigma_{a}^{+}\right]
+\mathbb{E}_{x}\left[\mathrm{e}^{-q \sigma_{0}^{-}}
\mathbb{E}\Big[\int_{0}^{\rho_{\mathrm{H}_{1}}^{+}}\mathrm{e}^{-q t}\mathrm{d}(-\underline{X}(t)\wedge0))\Big]
;\sigma_{0}^{-}<\sigma_{a}^{+}\right]
\nonumber\\
&&\hspace{-0.3cm}
+\mathbb{E}_{x}\left[\mathrm{e}^{-q \sigma_{0}^{-}}
\frac{r_{a}(\mathrm{H}_{1})}
{Z_{q}(\mathrm{H}_{1})}
;\sigma_{0}^{-}<\sigma_{a}^{+}\right]
\nonumber\\
\hspace{-0.3cm}&=&\hspace{-0.3cm}
\frac{1}{1-\ell}\int_{x}^{a}
\left(\frac{W_{q}(x)}
{W_{q}(s)}\right)^{\frac{1}{1-\ell}}
\left(Z_{q}(s)-\Psi^{\prime}(0+)W_{q}(s)-\frac{\overline{Z}_{q}(s)-\Psi^{\prime}(0+)
\overline{W}_{q}(s)}{W_{q}(s)}W_{q}^{\prime}(s)\right)\mathrm{d}s
\nonumber\\
\hspace{-0.3cm}&&\hspace{-0.3cm}
+\frac{1}{1-\ell}\int_{x}^{a}
\left(\frac{W_{q}(x)}
{W_{q}(s)}\right)^{\frac{1}{1-\ell}}
\left(
-\frac{\Psi^{\prime}(0+)}{q}+\frac{\overline{Z}_{q}(s)+\frac{\Psi^{\prime}(0+)}{q}}{Z_{q}(s)}
+\frac{r_{a}(s)}{Z_{q}(s)}
\right)
\nonumber\\
\hspace{-0.3cm}&&\hspace{0.3cm}
\times\left(\frac{W_{q}^{\prime}(s)}{W_{q}(s)}Z_{q}(s)-qW_{q}(s)\right)
\mathrm{d}s,\qquad x\in(0,a].
\end{eqnarray}
Differentiating both sides of \eqref{capt.inj.} with respect to $x$ gives the following differential equation
\begin{eqnarray}
\label{capt.inj.diff.}
r_{a}^{\prime}(x)
\hspace{-0.3cm}&=&\hspace{-0.3cm}
\frac{1}{1-\ell}\frac{W_{q}^{\prime}(x)}
{W_{q}(x)}r_{a}(x)
-
\frac{1}{1-\ell}
\left(Z_{q}(x)-\Psi^{\prime}(0+)W_{q}(x)-\frac{\overline{Z}_{q}(x)-\Psi^{\prime}(0+)
\overline{W}_{q}(x)}{W_{q}(x)}W_{q}^{\prime}(x)\right)
\nonumber\\
\hspace{-0.3cm}&&\hspace{-0.3cm}
-
\frac{1}{1-\ell}
\left(
-\frac{\Psi^{\prime}(0+)}{q}+\frac{\overline{Z}_{q}(x)+\frac{\Psi^{\prime}(0+)}{q}}{Z_{q}(x)}
+\frac{r_{a}(x)}{Z_{q}(x)}
\right)\left(\frac{W_{q}^{\prime}(x)}{W_{q}(x)}Z_{q}(x)-qW_{q}(x)\right)
\nonumber\\
\hspace{-0.3cm}&=&\hspace{-0.3cm}
\frac{1}{1-\ell}
\frac{qW_{q}(x)}{Z_{q}(x)}r_{a}\left(x\right)
-\frac{1}{1-\ell}\left(Z_{q}(x)-\left(\overline{Z}_{q}(x)+\frac{\Psi^{\prime}(0+)}{q}\right)
\frac{qW_{q}(x)}{Z_{q}(x)}\right),\quad x\in(0,a],
\end{eqnarray}
with boundary condition $r_{a}(a)=0$.
Solving \eqref{capt.inj.diff.} gives \eqref{e.d.a.t.01}. The proof is completed.
\end{proof}

\medskip
By Proposition \ref{pro4.1}, the expected total discounted tax payments minus the expected total discounted costs of capital injection is given by
\begin{eqnarray}
\label{psibar}
\overline{\psi}(x)\hspace{-0.3cm}&=&\hspace{-0.3cm}
g_{\infty}(x)-\varphi \,r_{\infty}(x)
\nonumber\\
\hspace{-0.3cm}&=&\hspace{-0.3cm}
-
\frac{\varphi}{1-\ell}\int_x^{\infty}
\left(Z_{q}(w)-\left(\overline{Z}_{q}(w)+\frac{\Psi^{\prime}(0+)}{q}\right)
\frac{qW_{q}(w)}{Z_{q}(w)}\right)
\left(\frac{Z_{q}(x)}{Z_{q}(w)}\right)^{\frac{1}{1-\ell}}\mathrm{d}w
\nonumber\\
\hspace{-0.3cm}&&\hspace{-0.3cm}
+\frac{\ell}{1-\ell}\int_x^{\infty}\left(\frac{Z_{q}(x)}{Z_{q}(w)}\right)^{\frac{1}{1-\ell}}\mathrm{d}w
,
\quad  x\in(0,\infty),\nonumber
\end{eqnarray}
where $\varphi\in(1,\infty)$ is the cost per unit amount of capital injected.
Define
\begin{eqnarray}
\label{v.expr.1}
\overline{\upsilon}(x):=
\overline{\psi}(x)
-\varphi\left(\overline{Z}_{q}(x)+\frac{\Psi^{\prime}(0+)}{q}\right),
\qquad  x\in(0,\infty).
\end{eqnarray}

If tax is collected after the risk process reaches the level $a\in(0,\infty)$, then
the expected total discounted delayed tax payments minus the expected total discounted cost of capital injection
can be expressed as
\begin{eqnarray}
\label{obj.func.}
\overline{\phi}(x;a)
\hspace{-0.3cm}&=&\hspace{-0.3cm}
\mathbb{E}_{x}\left[\mathrm{e}^{-q \rho_{a}^{+}}\right]
\overline{\psi}(a)
-\varphi \mathbb{E}_{x}\left[\int_{0}^{\rho_{a}^{+}}\mathrm{e}^{-q t}\mathrm{d}\left(-\underline{X}_{t}\wedge 0\right)\right]
\nonumber\\
\hspace{-0.3cm}&=&\hspace{-0.3cm}
\frac{Z_{q}(x)}{Z_{q}(a)}\overline{\psi}(a)
-\varphi \int_x^{a}\frac{Z_{q}(x)}{Z_{q}(w)}
\left(Z_{q}(w)-\left(\overline{Z}_{q}(w)+\frac{\Psi^{\prime}(0+)}{q}\right)
\frac{qW_{q}(w)}{Z_{q}(w)}\right)
\mathrm{d}w
\nonumber\\
\hspace{-0.3cm}&=&\hspace{-0.3cm}
\frac{Z_{q}(x)}{Z_{q}(a)}\overline{\psi}(a)
-\varphi\left(
\left(\overline{Z}_{q}(a)+\frac{\Psi^{\prime}(0+)}{q}\right)
\frac{Z_{q}(x)}{Z_{q}(a)}
-\left(\overline{Z}_{q}(x)+\frac{\Psi^{\prime}(0+)}{q}\right)\right)
\nonumber\\
\hspace{-0.3cm}&=&\hspace{-0.3cm}
\frac{Z_{q}(x)}{Z_{q}(a)}
\overline{\upsilon}(a)
+\varphi\left(\overline{Z}_{q}(x)+\frac{\Psi^{\prime}(0+)}{q}\right),
\end{eqnarray}
where we have used the Markov property, \eqref{two.side.exit.}, and \eqref{e.d.a.t.01} with $\ell=0$.

For the convenience of presentation, we define two auxiliary functions as
\begin{eqnarray}
\label{define.upsilon.1}
\overline{h}(a):=\overline{\upsilon}(a)-\overline{V}(a)\left(1-\varphi Z_{q}(a)\right)
\quad\mbox{and}\quad \overline{V}(b):=\frac{Z_{q}(b)}{Z_{q}^{\prime}(b)}.
\end{eqnarray}

\medskip
The following result characterizes the global maximum point and the global maximum value of the function $a\mapsto\overline{\phi}(x;a)$ for fixed $x$.

\medskip
\begin{thm}
\label{th.4.1}
When $\overline{\upsilon}(0)>\overline{V}(0)(1-\varphi Z_{q}(0))$, there is a unique positive solution $a^{+}\in(0,\infty)$ of the following equation
\begin{eqnarray}
\label{equ.det.argmax.1}
\hspace{-0.3cm}&&\hspace{-0.3cm}
\overline{h}(a)=0
\Leftrightarrow
\overline{\upsilon}^{\prime}(a)=1-\varphi Z_{q}(a),
\end{eqnarray}
and the function $\overline{\phi}(x;a)$ (of $a$) attains its global maximum at $a^{*}=a^{+}$.
Otherwise, when $\overline{\upsilon}(0)\leq \overline{V}(0)(1-\varphi Z_{q}(0))$, \eqref{equ.det.argmax.1} has no positive solution, and the function $\overline{\phi}(x;a)$ (of $a$) attains its global maximum at $a^{*}=0$.

Moreover, we have
\begin{eqnarray}
\label{}
\hspace{-0.3cm}&&\hspace{-0.3cm}
\overline{\phi}(x;a^{*})=\varphi\left(\overline{Z}_{q}(x)+\frac{\Psi^{\prime}(0+)}{q}\right)+\frac{Z_{q}(x)}{Z_{q}^{\prime}(a^{*})}\left(1-\varphi Z_{q}(a^{*})\right).
\nonumber
\end{eqnarray}
\end{thm}
\medskip

\begin{proof}
Firstly, we show a result that is useful in our main proof:
\begin{eqnarray}
\label{37}
\lim\limits_{a\rightarrow \infty}
\overline{h}(a)= \frac{\ell-1}{\Phi(q)}<0.
\end{eqnarray}
Using integration by parts, one can get
\begin{eqnarray}
\hspace{-0.3cm}&&\hspace{-0.3cm}
-\frac{1}{1-\ell}\int_a^{\infty}
\left(\overline{Z}_{q}(w)+\frac{\Psi^{\prime}(0+)}{q}\right)
\frac{qW_{q}(w)}{Z_{q}(w)}
\left(\frac{Z_{q}(a)}{Z_{q}(w)}\right)^{\frac{1}{1-\ell}}\mathrm{d}w
\nonumber\\
\hspace{-0.3cm}&=&\hspace{-0.3cm}
\left.\left(\overline{Z}_{q}(w)+\frac{\Psi^{\prime}(0+)}{q}\right)
\left(\frac{Z_{q}(a)}{Z_{q}(w)}\right)^{\frac{1}{1-\ell}}\right|_{w=a}^{\infty}
-
\int_a^{\infty}
Z_{q}(w)
\left(\frac{Z_{q}(a)}{Z_{q}(w)}\right)^{\frac{1}{1-\ell}}\mathrm{d}w
\nonumber\\
\hspace{-0.3cm}&=&\hspace{-0.3cm}
-\overline{Z}_{q}(a)-\frac{\Psi^{\prime}(0+)}{q}
-
\int_a^{\infty}
Z_{q}(w)
\left(\frac{Z_{q}(a)}{Z_{q}(w)}\right)^{\frac{1}{1-\ell}}\mathrm{d}w,
\quad  a\in(0,\infty),
\nonumber
\end{eqnarray}
and
\begin{eqnarray}
\label{}
\hspace{-0.3cm}&&\hspace{-0.3cm}-\frac{\ell}{1-\ell}\int_a^{\infty} Z_{q}(w)\left(\frac{Z_{q}(a)}{Z_{q}(w)}\right)^{\frac{1}{1-\ell}}\mathrm{d}w
\nonumber\\
\hspace{-0.3cm}&=&\hspace{-0.3cm}
\int_a^{\infty} \frac{Z_{q}(w)\left(Z_{q}(a)\right)^{\frac{1}{1-\ell}}}{Z_{q}^{\prime}(w)}
\mathrm{d}\left(Z_{q}(w)\right)^{-\frac{1}{1-\ell}+1}
\nonumber\\
\hspace{-0.3cm}&=&\hspace{-0.3cm}
-\overline{V}(a) Z_{q}(a)-
\int_a^{\infty}Z_{q}(w) \left(\frac{Z_{q}(w)}{Z_{q}^{\prime}(w)}\right)^{\prime}\left(\frac{Z_{q}(a)}{Z_{q}(w)}\right)^{\frac{1}{1-\ell}}
\mathrm{d}w
,
\quad  a\in(0,\infty).
\nonumber
\end{eqnarray}
Hence, by \eqref{v.expr.1} and \eqref{define.upsilon.1} we have
\begin{eqnarray}
\label{36}
\overline{h}(a)
\hspace{-0.3cm}&=&\hspace{-0.3cm}
\frac{\ell}{1-\ell}\int_a^{\infty}\left(\frac{Z_{q}(a)}{Z_{q}(w)}\right)^{\frac{1}{1-\ell}}\mathrm{d}w-\frac{Z_{q}(a)}{Z_{q}^{\prime}(a)}
\nonumber\\
\hspace{-0.3cm}&&\hspace{-0.3cm}
-
\varphi\int_a^{\infty}Z_{q}(w) \frac{q\left(W_{q}(w)\right)^{2}/W_{q}^{\prime}(w)-Z_{q}(w)}{q\left(W_{q}(w)\right)^{2}/W_{q}^{\prime}(w)}\left(\frac{Z_{q}(a)}{Z_{q}(w)}\right)^{\frac{1}{1-\ell}}
\mathrm{d}w.
\end{eqnarray}
Using \eqref{phi.0}, \eqref{phi.00} and L'H\^opital's rule, we get
\begin{eqnarray}
\lim\limits_{a\rightarrow \infty}\frac{\ell}{1-\ell}\int_a^{\infty}
\left(\frac{Z_{q}(a)}{Z_{q}(w)}\right)^{\frac{1}{1-\ell}}
\mathrm{d}w=\frac{\ell}{\Phi(q)},\nonumber
\end{eqnarray}
and
\begin{eqnarray}
\lim\limits_{a\rightarrow \infty}\int_a^{\infty}Z_{q}(w) \frac{q\left(W_{q}(w)\right)^{2}/W_{q}^{\prime}(w)-Z_{q}(w)}{q\left(W_{q}(w)\right)^{2}/W_{q}^{\prime}(w)}\left(\frac{Z_{q}(a)}{Z_{q}(w)}\right)^{\frac{1}{1-\ell}}
\mathrm{d}w=0,\nonumber
\end{eqnarray}
which, together with \eqref{phi.00} and \eqref{36}, gives \eqref{37}.\smallskip

Then, we show that the function $\bar{h}(a)$ can be used to find critical points of $a\mapsto \overline{\phi}(x;a)$.
By the expression of $\overline{\phi}(x;a)$ given by \eqref{obj.func.}, we have
\begin{eqnarray}
\label{daoshu.1}
\frac{\partial \overline{\phi}(x;a)}{\partial a}\hspace{-0.3cm}&=&\hspace{-0.3cm}\frac{Z_{q}(x)}{\left(Z_{q}(a)\right)^{2}}\left(Z_{q}(a)\overline{\upsilon}^{\prime}(a)-Z_{q}^{\prime}(a)
\,\overline{\upsilon}(a)\right)
\nonumber\\
\hspace{-0.3cm}&=&\hspace{-0.3cm}\frac{\ell}{1-\ell}\frac{Z_{q}(x)Z_{q}^{\prime}(a)}{\left(Z_{q}(a)\right)^{2}}
\left[\overline{\upsilon}(a)-\overline{V}(a)\left(1-\varphi Z_{q}(a)\right)\right]
\nonumber\\
\hspace{-0.3cm}&=&\hspace{-0.3cm}\frac{\ell}{1-\ell}\frac{Z_{q}(x)Z_{q}^{\prime}(a)}{\left(Z_{q}(a)\right)^{2}}
\overline{h}(a)
,
\end{eqnarray}
where the following expression for $\overline{\upsilon}^{\prime}(a)$ is employed
\begin{eqnarray}
\label{}
\overline{\upsilon}^{\prime}(a)
\hspace{-0.3cm}&=&\hspace{-0.3cm}
\frac{1}{1-\ell}\frac{\overline{\upsilon}(a)}{\overline{V}(a)}+\frac{\ell}{1-\ell}\left(\varphi Z_{q}(a)-1\right).
\nonumber
\end{eqnarray}
According to {(\ref{daoshu.1})}, $\frac{\partial \overline{\phi}(x;a)}{\partial a}=0$ is equivalent to \eqref{equ.det.argmax.1},
i.e., $a_{0}$ is a critical point of $a\mapsto \overline{\phi}(x;a)$ if $a_{0}$ solves \eqref{equ.det.argmax.1}.
To specify the nature of this critical point $a_{0}$, we examine the second partial derivative of $\bar{\phi}(x; a)$,
\begin{eqnarray}
\label{erjiedao.1}
\left.\frac{\partial^{2} \overline{\phi}(x;a)}{\partial a^{2}}\right|_{a=a_{0}}
\hspace{-0.3cm}&=&\hspace{-0.3cm}
\frac{\ell}{1-\ell}\frac{Z_{q}(x)}{Z_{q}(a_{0})}\left(\frac{1-\varphi Z_{q}(a_{0})}{Z_{q}^{\prime}(a_{0})}Z_{q}^{\prime\prime}(a_{0})+\varphi Z_{q}^{\prime}(a_{0})\right)
=\frac{\ell}{1-\ell}\frac{Z_{q}(x)}{Z_{q}(a_{0})}\bar{f}(a_0),
\end{eqnarray}
where
$$\overline{f}(a):=\frac{1-\varphi Z_{q}(a)}{Z_{q}^{\prime}(a)}Z_{q}^{\prime\prime}(a)+\varphi Z_{q}^{\prime}(a),\quad a\in(0,\infty).$$
Then, by \eqref{W}, we have
\begin{eqnarray}
\label{fbaryipie}
\overline{f}^{\prime}(a)=\left(1-\varphi Z_{q}(a)\right)\frac{W_{q}(a)W_{q}^{\prime\prime}(a)-\left(W_{q}^{\prime}(a)\right)^{2}}{\left(W_{q}(a)\right)^{2}}+\varphi qW'_q(a)>0,\quad a\in(0,\infty).
\end{eqnarray}
Hence, by \eqref{37} and a similar argument as that used in the proof of case {\bf A} of Theorem \ref{th.3.1}, one can verify that the map $a\mapsto \overline{\phi}(x;a)$ cannot have a local minimum point or saddle point in $[0,\infty)$.
Indeed,

\begin{itemize}

    \item if $a_{0}$ is a local minimum point of the map $a\mapsto \overline{\phi}(x;a)$, then there should be $a_{0}^{\prime}\in(a_{0},\infty)$ such that $\overline{f}(a_{0})>0$ and $\overline{f}(a_{0}^{\prime})<0$, which contradicts to \eqref{fbaryipie};

\item if $a_{0}$ is a saddle point of the map $a\mapsto \overline{\phi}(x;a)$,
then by \eqref{W} and
$$Z_{q}^{\prime\prime}(a_{0})=-\frac{\varphi \left(Z_{q}^{\prime}(a_{0})\right)^{2}}{1-\varphi Z_{q}(a_{0})},\qquad \overline{V}^{\prime}(a_{0})=\frac{1}{1-\varphi Z_{q}(a_{0})},$$
one can obtain $\overline{h}^{\prime}(a_{0})=0$ and
\begin{eqnarray}
\overline{h}^{\prime\prime}(a_{0})
\hspace{-0.3cm}&=&\hspace{-0.3cm}
\frac{Z_{q}^{\prime}(a_{0})Z_{q}^{\prime\prime\prime}(a_{0})\left(1-\varphi Z_{q}(a_{0})\right)^{2}-\varphi^{2}[Z_{q}^{\prime}(a_{0})]^{4}}
{\left(1-\varphi Z_{q}(a_{0})\right)[Z_{q}^{\prime}(a_{0})]^{3}/Z_{q}(a_{0})}
\nonumber\\
\hspace{-0.3cm}&>&\hspace{-0.3cm}
\frac{\left(Z_{q}^{\prime\prime}(a_{0})\right)^{2}\left(1-\varphi Z_{q}(a_{0})\right)^{2}-\varphi^{2}[Z_{q}^{\prime}(a_{0})]^{4}}
{\left(1-\varphi Z_{q}(a_{0})\right)[Z_{q}^{\prime}(a_{0})]^{3}/Z_{q}(a_{0})}=0,
\nonumber
\end{eqnarray}
implying that $\overline{h}(a)$ reaches a local minimum value of $0$ at $a_{0}$. In this case there must exist $a_{0}^{\prime}\in(a_{0},\infty)$ such that $a_{0}$ is a maximum point or another saddle point of the map  $a\mapsto\overline{\phi}(x;a)$, i.e.,
$\overline{f}(a_{0})=0$
and $\overline{f}(a_{0}^{\prime})\leq0$, which again contradicts to \eqref{fbaryipie}.
\end{itemize}
We thus conclude that the equation \eqref{equ.det.argmax.1} has at most one positive solution.\smallskip

Finally, we shall examine the existence of any positive solution of the equation \eqref{equ.det.argmax.1}.
\begin{itemize}
    \item When $\overline{\upsilon}(0)>\overline{V}(0)(1-\varphi Z_{q}(0))$, by \eqref{37} one concludes that a positive solution $a^{+}\in(0,\infty)$ of \eqref{equ.det.argmax.1} exists such that
\begin{eqnarray}
\label{maximumf(b00)01}
\hspace{1cm}\overline{\upsilon}(a)>\overline{V}(a)(1-\varphi Z_{q}(a)),\,\,a\in[0,a^{+})\quad \mbox{and} \quad \overline{\upsilon}(a)<\overline{V}(a)(1-\varphi Z_{q}(a)),\,\,a\in(a^{+},\infty),\nonumber
\end{eqnarray}
i.e., the map $a\mapsto\overline{\phi}(x;a)$ reaches its global maximum value at $a^{*}=a^{+}$.

\item When $\overline{\upsilon}(0)<\overline{V}(0)(1-\varphi Z_{q}(0))$, then \eqref{equ.det.argmax.1} has no positive solution and
\begin{eqnarray}
\label{maximumf(b00)02}
\overline{\upsilon}(a)<\overline{V}(a)(1-\varphi Z_{q}(a)),\qquad a\in[0,\infty),
\end{eqnarray}
i.e., the map $a\mapsto\overline{\phi}(x;a)$ reaches its global maximum value at $a^{*}=0$.

\item When $\overline{\upsilon}(0)=\overline{V}(0)(1-\varphi Z_{q}(0))$, \eqref{maximumf(b00)02} holds true for all $a\in(0,\infty)$ and the map $a\mapsto\overline{\phi}(x;a)$ reaches its global maximum value at $a^{*}=0$.
\end{itemize}
The proof is completed.
\end{proof}

\begin{rem}
We would like to emphasize that, in Theorem \ref{th.3.1} the characterization of the optimal tax implementation threshold $b^{*}$ relies on the assumption of the existence of a completely monotone density of the L\'evy measure of $X$, whilst in Theorem \ref{th.4.1}, when capital injection is involved, no such assumption is needed to obtain the criterion of the optimal tax implementation threshold $a^{*}$.
\end{rem}

\section{Numerical examples on optimal delay of taxation implementation}

\medskip
In this section we consider a Cramer-Lundberg risk model example. Let $X(t)=x+ ct-S(t)$ with $x\ge 0$, and $S(t)$ is a compound Poisson process with rate $\lambda>0$ and an exponential jump distribution $F(x)=1-\mathrm{e}^{-\mu x}$, $\mu>0$. It can be verified that $X$ has a scale function
\begin{eqnarray}
W_q(x)=\frac{A_1(q)}{c}\mathrm{e}^{\theta_1(q)x}-\frac{A_2(q)}{c}\mathrm{e}^{\theta_2(q)x},\qquad x\geq0,\,q\geq0,\label{Ex2_W}
\end{eqnarray}
where $\kappa(q)=\sqrt{(c\mu-\lambda-q)^2+4cq\mu}$, $A_1(q)=\frac{\mu+\theta_1(q)}{\theta_1(q)-\theta_2(q)}=\frac{\lambda+q+c\mu}{2\kappa(q)}+\frac{1}{2}$, $A_2(q)=\frac{\mu+\theta_2(q)}{\theta_1(q)-\theta_2(q)}=\frac{\lambda+q+c\mu}{2\kappa(q)}-\frac{1}{2}$, and
\begin{eqnarray*}
\hspace{-0.2cm}&&\hspace{-0.2cm}\theta_1(q)=\frac{\lambda+q-c\mu+\kappa(q)}{2c},\qquad \theta_2(q)=\frac{\lambda+q-c\mu-\kappa(q)}{2c}.
\end{eqnarray*}
Further, for $x, q\ge 0$,
\begin{eqnarray*}
Z_q(x)\hspace{-0.2cm}&=&\hspace{-0.2cm} 1+q\int^x_0\Big[\frac{A_1(q)}{c}\mathrm{e}^{\theta_1(q)y}-\frac{A_2(q)}{c}\mathrm{e}^{\theta_2(q)y}\Big]\mathrm{d}y\\
\hspace{-0.2cm}&=&\hspace{-0.2cm} \frac{qA_1(q)}{c\theta_1(q)}\mathrm{e}^{\theta_1(q)x}-\frac{qA_2(q)}{c\theta_2(q)}\mathrm{e}^{\theta_2(q)x}.
\end{eqnarray*}
Without loss of generality, we let $c=1.2$, $\mu=1$, $\lambda=1$, and $q=0.05$. In the following we shall present two numerical examples which illustrate the previously discussed two optimal delay of taxation implementation models respectively. \medskip

{\bf Example 1.} We consider the optimal delay of taxation implementation with a terminal value at ruin first. For the purpose of comparison, we consider three levels of tax rate, i.e., $\ell=0.1, 0.2$ and 0.3. To get a concrete idea on how the condition given in Theorem 3.1, i.e., $v(0)>V(0)(1-SqW_q(0))$, behaves for varying $q$ and $S$ values, we show a 3D graph with $\ell=0.1$, $S\in [-5,+5]$ and $q\in (0, 0.05]$ in Figure 1. It can be seen that when $q$ is extremely small ($q<0.004$), the condition is met for certain ranges of positive $S$ values. Otherwise, no positive terminal value can satisfy the condition. It implies that leaving a positive tax payment at the end cannot pair with taxation delay to achieve an maximal total expected discounted tax payments until ruin at normal rates of discounting.

\begin{figure}[H]
\begin{center}
\includegraphics[width=5in]{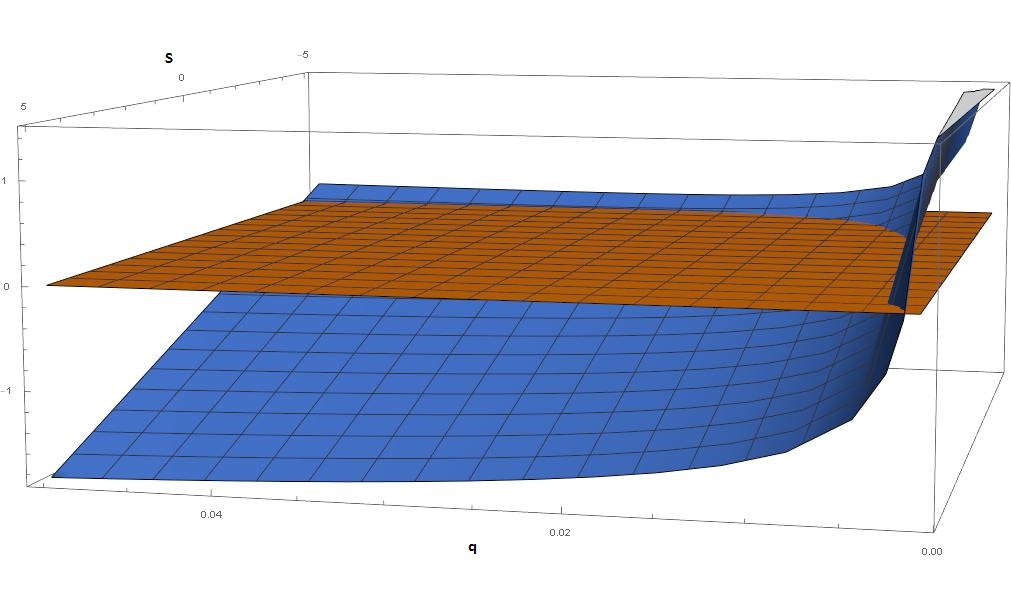}
\caption{Optimal taxation implementation delay for various $S$ values}
\label{figure2}
\end{center}
\end{figure}

According to the condition of positive $b^*$ given in Theorem 3.1, we calculate the following ranges of $S$ values at each assumed $\ell$ level, which are summarized in Table 1 and Table 2. Further, Figure 2 and Figure 3 contain corresponding $b^*$ values for specifically chosen ranges for $S$. Table 1 and Figure 2 confirm the fact that it is worth implementing tax delays only when $S<0$ at $q=0.05$, i.e. the insurance company receives a tax benefit $S$ at ruin. From Table 2 and Figure 3 one can see that when $q=0.002$, both negative and certain ranged positive $S$ values lead to positive $b^*$ which means the delay in taxation implementation is worthwhile. In general, the optimal tax threshold level $b^*$ is a decreasing function of $S$. Also, given $S$, $b^*$ is an increasing function of $\ell$.

\begin{table}[ht]
\centering\caption{Existence of positive $b^*$ at $q=0.05$} \vspace{-.2cm}
\begin{tabular}{| c | c | c | c | }   \hline
$\ell$ &$v(0)$ &$V(0)(1-S q W_q(0))$   &$b^*>0$\\ \hline
0.1 &$0.29630-0.24801 S$    &$1.14286-0.04762 S$    &$S<-4.22$\\
0.2 &$0.55487-0.21900 S$      &$1.14286-0.04762 S$    &$S<-3.43$\\
0.3 &$0.77143-0.18931 S$    &$1.14286-0.04762 S$    &$S<-2.62$\\ \hline
\end{tabular}\label{table-2}
\end{table}

\begin{table}[ht]
\centering\caption{Existence of positive $b^*$ at $q=0.002$} \vspace{-.2cm}
\begin{tabular}{| c | c | c | c | }   \hline
$\ell$ &$v(0)$ &$V(0)(1-S q W_q(0))$   &$b^*>0$\\ \hline
0.1 &$1.75676-0.144583 S$    &$1.1976-0.001996 S$    &$S<3.92$\\
0.2 &$2.87400-0.11451 S$      &$1.1976-0.001996 S$    &$S<14.9$\\
0.3 &$3.36578-0.08521 S$    &$1.1976-0.001996 S$    &$S<26.06$\\ \hline
\end{tabular}\label{table-2}
\end{table}

\begin{figure}[H]
\begin{center}
\includegraphics[width=5in]{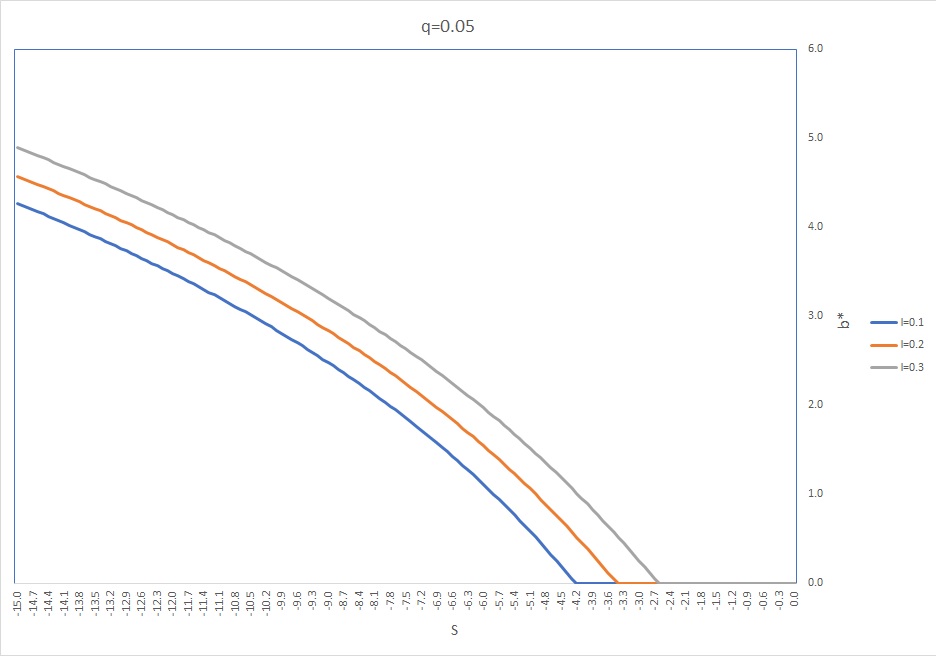}
\caption{Optimal taxation implementation delay for various $S$ values}
\label{figure2}
\end{center}
\end{figure}

\begin{figure}[H]
\begin{center}
\includegraphics[width=5in]{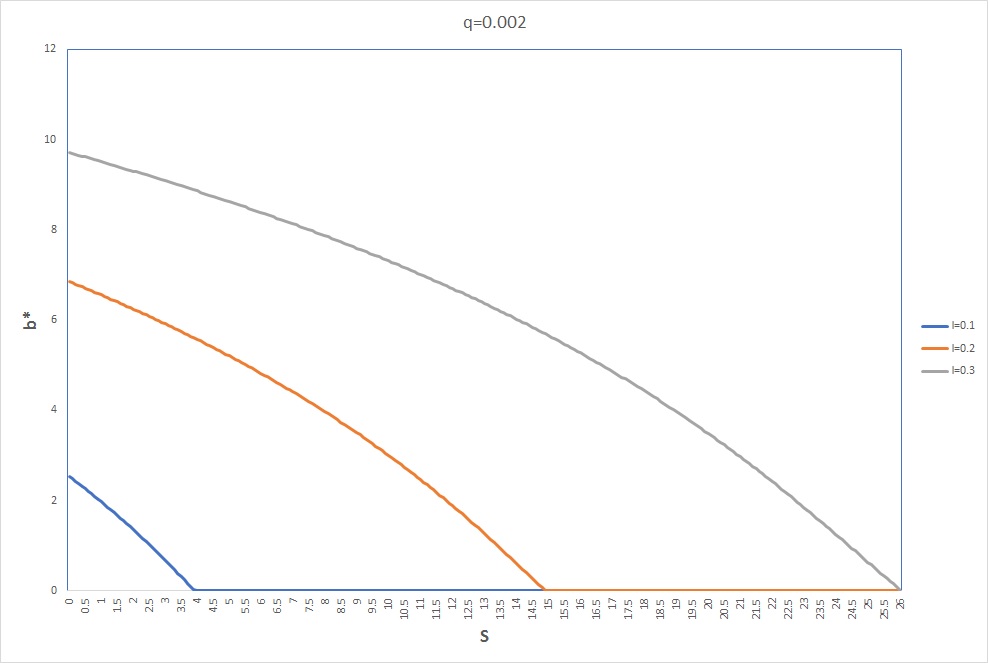}
\caption{Optimal taxation implementation delay for various $S$ values}
\label{figure2}
\end{center}
\end{figure}

{\bf Example 2.} We consider the optimal delay of taxation implementation with capital injections in this example. Similarly, we shall consider three tax rates, $\ell =0.1, 0.2$, and 0.3, for comparison purpose. Let $q=0.05$. According to the condition of positive $a^*$ given in Theorem 4.1, we calculate the following ranges of $\varphi$ values at each assumed $\ell$ level, which are summarized in Table 3. Further, Figure 4 shows the corresponding optimal taxation threshold levels for a specifically chosen range for $\varphi$. One can see from Table 3 and Figure 4 that it is worth implementing tax delays only when $\varphi>1$, which coincides with the definition of $\varphi$. The optimal tax threshold level $a^*$ is an increasing function of the cost of capital injection per dollar amount. Also, for the same cost of capital injection per dollar amount, a higher tax rate leads to a higher optimal tax threshold level.

\begin{table}[ht]
\centering\caption{Existence of positive $a^*$ in various cases} \vspace{-.2cm}
\begin{tabular}{| c | c | c | c | }   \hline
$\ell$ &$\bar{v}(0)$ &$\bar{V}(0)(1-\varphi Z_q(0))$   &$a^*>0$\\ \hline
0.1 &$-2.9622-4\varphi$    &$24-24\varphi$    &$\varphi>1.348$\\
0.2 &$-2.45091-4\varphi$      &$24-24\varphi$    &$\varphi>1.323$\\
0.3 &$-1.93869-4\varphi$    &$24-24\varphi$    &$\varphi>1.297$\\ \hline
\end{tabular}\label{table-2}
\end{table}

\begin{figure}[H]
\begin{center}
\includegraphics[width=5in]{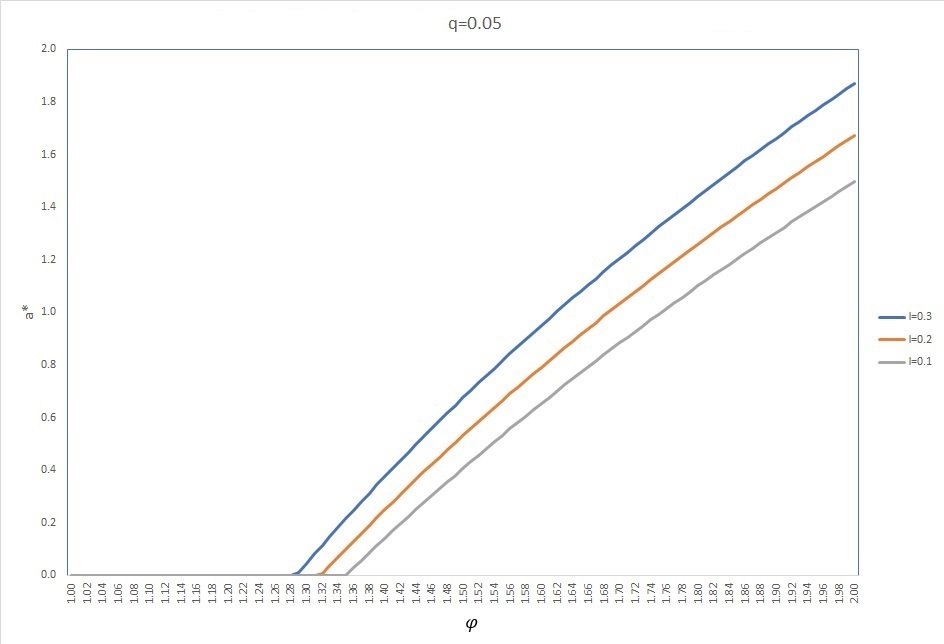}
\caption{Optimal taxation implementation delay for various $\varphi$ values}
\label{figure2}
\end{center}
\end{figure}

\begin{rem}
A negative terminal value of $S$, i.e., a benefit to the insurance company at ruin, can generally trigger a positive optimal tax implementation threshold level $b^*$, but a positive terminal value of $S$, eg a claw-back of early tax relief from the insurance company at ruin can only be effective when the discount rate is extremely small.
\end{rem}

\begin{rem}
A main cause of claiming a tax benefit by the insurance company at ruin is due to the capital loss, which is obvious, according to the capital gain tax rule. The higher the initial tax implementation threshold is, the higher the potential tax benefit can be claimed at the end. It is consistent with our findings in Example 1.
\end{rem}

\begin{rem}
The capital injection strategy seems more robust in respect of discount rates than the terminal value strategy. It is mainly because the capital injections could be needed at a reasonably early stage and the capital injections spread out the whole time line, whilst the terminal value can only occur at the very end, i.e. the time of ruin.
\end{rem}

\medskip
\section*{Acknowledgements}
Wenyuan Wang is very grateful to The University of Melbourne where part of the work on this paper was completed during his visit in 
 2018.


\end{document}